\documentclass{article}
\RequirePackage{etex}
\usepackage{xcolor}
\usepackage{graphicx} 

\usepackage{amsmath}
\usepackage{amsfonts,amsmath,amsthm,amssymb}
\usepackage{mathtools}
\usepackage[numbers,sort,compress]{natbib}
\usepackage{geometry}
\usepackage{pgf,pgfplots}
\usepackage{caption}
\usepackage{subcaption}
\usepackage[]{enumitem}
\usepackage{bigints}
\usepackage{float}
\usepackage[title]{appendix}

\usepackage{enumitem}

\usepackage{bm}

\RequirePackage[colorlinks,citecolor=blue,urlcolor=blue]{hyperref}

\usepackage{autonum}
\theoremstyle{plain}
\newtheorem{lem}{Lemma}

\newtheorem{prop}[lem]{Proposition}
\newtheorem{thm}[lem]{Theorem}

\theoremstyle{definition}
\newtheorem{remark}[lem]{Remark}
\newtheorem{defn}[lem]{Definition}

\newtheorem{assum}[lem]{Assumption}

\newcommand{\R}{\mathbb{R}}
\renewcommand{\P}{\mathbb{P}}

\newcommand{\E}{\mathbb{E}}
\newcommand{\F}{\mathcal{F}}
\newcommand{\N}{\mathbb{N}}

\newcommand{\Bcal}{{\mathcal B}}

\newcommand{\Dcal}{{\mathcal D}}
\newcommand{\Ecal}{{\mathcal E}}
\newcommand{\Fcal}{{\mathcal F}}

\newcommand{\Ncal}{{\mathcal N}}

\newcommand{\Pcal}{{\mathcal P}}

\newcommand{\Wcal}{{\mathcal W}}

\numberwithin{lem}{section}

\makeatletter
\newcommand{\oset}[3][0.6ex]{%
  \mathrel{\mathop{#3}\limits^{
    \vbox to#1{\kern-2\ex@
    \hbox{$\scriptstyle#2$}\vss}}}}
\makeatother

\makeatletter
\renewenvironment{proof}[1][\proofname] {\par\pushQED{\qed}\normalfont\topsep6\p@\@plus6\p@\relax\trivlist\item[\hskip\labelsep\bfseries#1\@addpunct{.}]\ignorespaces}{\popQED\endtrivlist\@endpefalse}
\makeatother
\title{Stochastic factors can matter: \\ improving robust growth under ergodicity}
\author{B\'alint Binkert\footnote{Department of Mathematics, ETH Zurich, \href{mailto:balint.binkert@gmail.com}{balint.binkert@gmail.com}}, David Itkin\footnote{Department of Statistics, London School of Economics and Political Science, \href{mailto:d.itkin@lse.ac.uk}{d.itkin@lse.ac.uk}}, Paul Mangers Bastian\footnote{Department of Statistics, London School of Economics and Political Science, \href{mailto:P.J.Mangers-Bastian@lse.ac.uk}{p.j.mangers-bastian@lse.ac.uk}}, Josef Teichmann\footnote{Department of Mathematics, ETH Zurich, \href{mailto:josef.teichmann@math.ethz.ch}{josef.teichmann@math.ethz.ch}}}

\pgfplotsset{compat=1.18}
\begin{document}

\maketitle

\begin{abstract}
    Drifts of asset returns are notoriously difficult to model accurately and, yet, trading strategies obtained from portfolio optimization are very sensitive to them. To mitigate this well-known phenomenon we study robust growth-optimization in a high-dimensional incomplete market under drift uncertainty of the asset price process $X$, under an additional ergodicity assumption, which constrains but does not fully specify the drift in general. The class of admissible models allows $X$ to depend on a multivariate stochastic factor $Y$ and fixes (a) their joint volatility structure, (b) their long-term joint ergodic density and (c) the dynamics of the stochastic factor process $Y$. A principal motivation of this framework comes from pairs trading, where $X$ is the spread process and models with the above characteristics are commonplace. Our main results determine the robust optimal growth rate, construct a worst-case admissible model and characterize the robust growth-optimal strategy via a solution to a certain partial differential equation (PDE). We demonstrate that utilizing the stochastic factor leads to improvement in robust growth complementing the conclusions of the previous study \cite{itkin2025ergodic}, which additionally robustified the dynamics of the stochastic factor leading to $Y$-independent optimal strategies. Our analysis leads to new financial insights, quantifying the improvement in growth the investor can achieve by optimally incorporating stochastic factors into their trading decisions. We illustrate our theoretical results on several numerical examples including an application to pairs trading. 
\end{abstract}





\section{Introduction}
In this paper we study an asymptotic growth-optimization problem under model uncertainty and ergodicity. Our focus is on an investor who seeks stability in (discounted) asset prices and trades on this stability persisting, which is commonplace in pairs trading and certain statistical arbitrage strategies. In practice, investors estimate asset volatilities, distributions of asset returns and obtain noisy factors that provide partial information about price movements. However, direct estimation of the drifts of tradeable securities is usually inaccurate, due to the low signal-to-noise ratios present in financial data. Postulating a parametric model for asset returns and estimating only a few select parameters that pin down the drift process may seem plausible in some cases, but for the purpose of portfolio optimization directly leads to very strong structural assumptions on the form of the ensuing optimal strategy. For instance, in the case of $d$ risky assets $dX_{t,i} = \mu_i(X_t,Y_t)dt + \sum_{j=1}^d\sigma_{ij}(X_t,Y_t) dW_{t,j}$ depending on a stochastic factor $Y$ and a risk-free asset normalized to one, the growth-optimal holdings (under full information) are 
\begin{equation} \label{eqn:growth_optimal}
    \theta_t = (\sigma(X_t,Y_t) \sigma(X_t,Y_t)^\top)^{-1}\mu(X_t,Y_t),
\end{equation} which depend in a linear way on the chosen parametric drift specification.  

To this end, we study a robust growth-optimization problem with drift uncertainty,
\begin{equation} \label{eqn:lambda_P}
    \lambda_{\Pcal} = \sup_{\theta \in \Theta}\inf_{\P \in \Pcal} g(\theta;\P).
\end{equation}
in an incomplete market setup incorporating an $m$-dimensional stochastic factor $Y$, which is not traded but influences the dynamics of $X$. Here, $\Theta$ is the set of all admissible strategies modelling full information on $(X,Y)$, $\Pcal$ denotes the class of models we robustify over, which consists of all admissible probability measures $\P$ governing the dynamics of $(X,Y)$, and $g(\theta;\P)$ is the investor's asymptotic growth rate when using the strategy $\theta$ under the law $\P$. The class $\Pcal$ is specified using three inputs constraining them; a matrix-valued function $c(x,y)$ which specifies the joint volatility structure of $(X,Y)$, a positive function $p(x,y)$  which encodes the aforementioned price stability by being the long-run ergodic density of $(X,Y)$ and a vector valued function $b_Y(x,y)$ specifying the drift of the stochastic factor $Y$ (the precise definitions of these quantities are in Section~\ref{sec:problem_formulation}). As such, only the drift of $X$ is parametrically unspecified in this framework, but of course constrained by the inputs.

This setup builds on the previous papers \cite{kardaras2021ergodic,itkin2025ergodic} where similar problems were studied, but for different classes of admissible measures. Indeed, \cite{kardaras2021ergodic} studied this problem in the complete market setup, without a stochastic factor process $Y$, while \cite{itkin2025ergodic} studied the incomplete market case but also robustified over the drift of $Y$. The latter study found that the optimal strategy when robustifying over the drifts of both $X$ and $Y$ is independent of $Y$. This indicates that the class of measures which encodes uncertainty for the drifts of both $X$ and $Y$ is so large, that it admits an adversarial worst-case measure under which the stochastic factor $Y$ becomes superfluous. The main new feature in this work is allowing for a third input $b_Y$, which encodes additional information about the stochastic factor process $Y$ and which puts a further constraint on admissible measures. The results of this work demonstrate that in this setting the robust growth-optimal strategy depends on both $X$ and $Y$ showing that stochastic factors can indeed improve robust growth. However, if the investor is mistaken and the true drift of $Y$ differs from what is assumed, the investor opens themselves up to underperformance relative to the optimal strategy from \cite{itkin2025ergodic}; see Section~\ref{sec:discussion} for a detailed discussion.

In practice, investors may be more confident in their estimates for the dynamics of stochastic factors than for asset returns. Indeed, certain factors may have higher signal-to-noise ratios than asset returns and have additional data available to perform statistical estimation. An example of this type includes stochastic volatility, for which derivatives data and volatility indices, such as the VIX, provide additional data for estimation, in addition to price data. In other cases factors may be exogenously fixed by the investor leading to known dynamics by construction. Additionally, even in cases when the investor's estimates for the drift of $Y$ are relatively low-confidence, it may be of interest to study the different conclusions and optimal strategies specified under the different input frameworks. Indeed, the functional form of the optimal strategy \eqref{eqn:growth_optimal} is the same regardless of any assumed dynamics for $Y$, but the holdings distribution depends on the stochastic factor's law. 
This leads to a more subtle and complex relationship between the drift of $Y$ and the growth-optimal strategy than between the drift of $X$ and the growth-optimal strategy. The robust framework of this paper allows us to better understand this relationship. Another important message of our analysis: stochastic factors matter if one has good knowledge on their dynamics.

Our main results in Section~\ref{sec:main_results} solve, in a general high-dimensional market, for the robust growth rate $\lambda_{\Pcal}$ and characterize the robust growth-optimal strategy $\theta^*$, which is specified by a feedback form function $\phi^*(x,y)$. Our approach uses the calculus of variations to tackle the optimization problem and leads us to Euler--Lagrange partial differential equations (PDEs) which we prove $\phi^*$ satisfies. The PDE depends only on derivatives in $x$, so that $y$ can be treated as a parameter; that is, every state of the factor process has its own associated PDE that specifies the optimal strategy one should use when $Y_t = y$. 
We then apply our general framework to several examples including a high-dimensional Gaussian specification and an extended look at a pairs trading application. In the context of pairs trading, we robustify the widely used Central Tendency Ornstein--Uhlenbeck (CTOU) model (see e.g.\ \cite{leung2018optimal,liu2017intraday}) and explore extensions that incorporate fat-tailed return distributions and stochastic volatility.

The paper is organized as follows. Section~\ref{sec:setup} rigorously introduces the setup and the robust optimization problem. Section~\ref{sec:heuristic} then discusses the heuristic approach and the main ideas for solving the problem. Our approach extends the techniques used in \cite{kardaras2021ergodic,itkin2025ergodic}, which connects the Euler--Lagrange equation coming from the optimization problem to the Fokker--Planck equation describing the law of $(X,Y)$, to incorporate trading strategies that depend on both $X$ and $Y$. Section~\ref{sec:rigorous} then formulates the rigorous mathematical assumptions under which our results are proven. Our Assumption~\ref{ass:conditions} relaxes the assumptions required in \cite{itkin2025ergodic}, allowing for unbounded factor processes natural for many examples and reducing the number of integrability conditions that need to be satisfied. All of our main results, characterizing the robust optimal strategy, robust growth rate and the worst-case measure are stated in Section~\ref{sec:main_results} with proofs postponed to Appendix~\ref{sec:proofs} for better readability. A discussion comparing our results to \cite{itkin2025ergodic} and highlighting key financial insights is carried out in Section~\ref{sec:financial_insights}. Section~\ref{sec:examples} then applies our framework to several examples. In Section~\ref{sec:OU} we study a high-dimensional Gaussian environment, where all of our assumptions are carefully checked, while in Section~\ref{sec:pairs_trading} we take an extended look at a pairs trading application. Our theoretical results are complemented by numerical experiments, which demonstrate the conclusions of our study in stylized, but representative, market environments. Section~\ref{sec:conclusion} concludes and discusses directions for future work.

\section{Setup} \label{sec:setup}
\subsection{Problem formulation} \label{sec:problem_formulation}
We work with a financial market that contains a risk-free numeraire asset, which is normalized to one, and $d \geq 1$ risky assets. Notice that the numeraire of our investment universe is chosen in such a way that ergodicity assumptions can reasonably hold true. We assume that the risky asset price process $X = (X_1,\dots,X_d)$ takes values in an open connected set $E \subset \R^d$. The process $X$ will depend on an $m$-dimensional stochastic factor process $Y$ taking values in a connected open set $D \subset \R^m$ for some $m \geq 1$. We set $F = E \times D$ and in this paper will generically denote elements of $F$  by $z = (x,y)$ for $x \in E$ and $y \in D$. Similarly, gradients of a function of $z$ taken in only the $x$-variable will be denoted by $\nabla_x$ and in only the $y$-variable by $\nabla_y$ and the full gradient denoted, as usual, by $\nabla$. For a vector $v$ of size $d+m$, we will write $v_X$ and $v_Y$ for the vector consisting of the first $d$ and final $m$ components of $v$ respectively.

We take a triple $(c,p,b_Y)$ as inputs to the problem for functions $c:F \to \mathbb{S}^{d+m}_{++}$, $p:F\to (0,\infty)$ and $b_Y:F \to \R^m$. Here, for any $n \in \N$, $\mathbb{S}^{n}_{++}$ is the cone of symmetric positive definite matrices of size $n \times n$. We will canonically write $c$ in block form as \[c(z) = \begin{bmatrix}
    c_X(z) & c_{XY}(z) \\ c_{YX}(z) & c_Y(z)
\end{bmatrix},\] where $c_X(z) \in \mathbb{S}^{d}_{++}$, $c_Y(z) \in \mathbb{S}^m_{++}$ and $c_{XY}(z) = c_{YX}^\top (z)$ is a matrix of size $d \times m$. We make the following assumptions on the regularity of the inputs. 
\begin{assum}
\label{ass:inputs}
There exists $\gamma \in (0,1]$ such that 
\begin{enumerate}[noitemsep] 
\item $c \in C^{2,\gamma}(F;\mathbb{S}^{d+m}_{++})$,
\item $p \in C^{2,\gamma}(F;(0,\infty))$ is such that $\int_F p(z)dz = 1$,
\item $b_Y \in C^{1,\gamma}(F;\R^m)$.
\end{enumerate}
\end{assum} 


Regarding the probabilistic structure, we will work on the canonical path space $\Omega = C([0,\infty);F)$ with Borel $\sigma$-algebra induced by the topology of local uniform convergence. The coordinate process is denoted by $Z = (X,Y)$ and we let $(\Fcal_t)_{t \geq 0}$ by the right-continuous enlargement of the filtration generated by $Z$ modelling full information on $(X,Y)$. On this space we will consider a class of probability measures under which $Z$ has quadratic variation prescribed by $c$, ergodic behaviour prescribed by $p$ and the drift of the $Y$-component of $Z$ is given in terms of $b_Y$.
\begin{defn}[Admissible class of measures] \label{defn:class} Given inputs $(c,p,b_Y)$ satisfying Assumption~\ref{ass:inputs}, we define a class of probability measures $\Pcal$ on $(\Omega,\F)$ consisting of all measures $\P$ under which
\begin{enumerate}
    \item \label{item:dynamics} $Z$ is a continuous semimartingale with dynamics 
    \begin{equation} \label{eqn:Z_dynamics} 
    dZ_t  = dA_t^\P + c^{1/2}(Z_t)dW_t,
    \end{equation}
    where $W$ is a standard $(d+m)$-dimensional Brownian motion, $c^{1/2}(z)$ is a matrix square root of $c(z)$ and $A^\P$ is a finite variation process of the form
    \begin{equation} \label{eqn:dA}
        dA_t^\P = \begin{pmatrix} dA^\P_{X,t} \\ c_Y(Z_t)b_Y(Z_t)dt \end{pmatrix},
    \end{equation}
    where $A^\P_X$ is some continuous adapted $d$-dimensional process of finite variation,
    \item $Z$ satisfies the ergodic property,
    \begin{equation} \label{eqn:ergodic}
    \lim_{T \to \infty} \frac{1}{T}\int_0^T h(Z_t)dt = \int_F h(z)p(z)dz; \qquad \P\text{-a.s.,}
    \end{equation}
    for every locally bounded $h:F \to \R$ with $\int_F h(z)p(z)dz < \infty$. \label{item:ergodic}
\end{enumerate}
\end{defn}

As in \cite{kardaras2021ergodic} and \cite[Section~5]{itkin2025ergodic}, the diffusion matrix of the coordinate process is specified by the matrix $c$ and the density $p$ governs the long-term behaviour of the coordinate process. However, differently from \cite{kardaras2021ergodic} the market is incomplete and, differently from \cite{itkin2025ergodic}, we assume that the local dynamics of the stochastic factor process are entirely known. Indeed, the drift of $Y$ is specified by the input $b_Y$, which was not present in previous studies.\footnote{The parametrization $c_Y(z)b_Y(z)$ for the drift of $Y$ in \eqref{eqn:dA} is convenient for later computations, but does not amount to any additional structural condition on the drift of $Y$ since $c_Y$ is everywhere invertible.} Accordingly, the class $\Pcal$ is a subset of the class considered in \cite{itkin2025ergodic}, which did not restrict the drift coefficients of $Y$.\footnote{The paper \cite{itkin2025ergodic} introduced a family of classes $\Pi_K$ for so-called \emph{$K$-modifications}. These are not needed in this paper and, hence, for simplicity we drop the subscript $K$ when comparing to the classes from \cite[Section~5.2]{itkin2025ergodic}.}
As such, the class $\Pcal$ provides uncertainty over only the drift of $X$, while the classes of measures in \cite{itkin2025ergodic} additionally encoded uncertainty over the drift of $Y$.


With the class of measures fixed, we now turn our attention to the optimal investment criterion. The set of admissible strategies is denoted by $\Theta$ and consists of all predictable $d$-dimensional processes $(\theta_t)_{t \geq 0}$ modelling full information on $(X,Y)$, which are additionally $X$-integrable with respect to every measure $\P \in \Pcal$. 
When the investor uses a strategy $\theta \in \Theta$ their wealth process is given by
\begin{equation} \label{eqn:wealth} 
V_T^\theta = \Ecal\bigg(\int_0^T \theta_t^\top dX_t\bigg), \qquad T \geq 0,
\end{equation}
where $\Ecal$ denotes stochastic exponentiation and we assume without loss of generality that the initial wealth is normalized to $V_0^\theta = 1$.

We consider the asymptotic growth rate as the optimality criterion. For $\theta \in \Theta$ and $\P \in \Pcal$ this quantity is defined as 
\begin{equation}\label{eqn:growth_rate}
    g(\theta;\P) = \sup\left\{\gamma \in \R: \liminf_{T \to \infty} \frac{1}{T}\log V_T^\theta \geq \gamma; \quad \P\text{-a.s.}\right\}
\end{equation}
The corresponding robust growth rate is then given by \eqref{eqn:lambda_P}.
Our ambitious goal in the sequel is to characterize the growth rate $\lambda_\Pcal$ in terms of the inputs, find the optimal strategy $\theta^*$ achieving it and compare the robust optimal growth rate and strategy to the one previously obtained in \cite{itkin2025ergodic}. It is remarkable how explicit the results are.

\begin{remark} \label{rem:growth_rate}
The pathwise growth rate definition \eqref{eqn:growth_rate} was previously studied in  \cite{kardrarasrobust2012} and \cite[Section~2.3.7]{karatzas2021portfolio}, but differs from the in-probability definition
\begin{equation} \label{eqn:growth_rate_prob}
g_{\mathrm{prob}}(\theta;\P) = \sup\Big\{\gamma \in \R: \lim_{T \to \infty} \P\Big(\frac{1}{T}\log V_T^\theta \geq \gamma\Big) = 1\Big\}
\end{equation}
used in \cite{kardaras2021ergodic,itkin2025ergodic}. Our approach below allows us to characterize the robust growth rate when using the definition \eqref{eqn:growth_rate} and since $g(\theta;\P) \leq g_{\mathrm{prob}}(\theta;\P)$ we prefer to use it here as it is a more conservative choice. However, our results carry over to the in-probability notion of growth rate as well. Moreover, the results of \cite{kardaras2021ergodic,itkin2025ergodic} can be extended to the growth-rate definition \eqref{eqn:growth_rate} when restricting the admissible class of measures in those papers to those with finite asymptotic growth; see Definition~\ref{defn:P0} below for this condition and the proof of Theorem~\ref{thm:main} in Appendix~\ref{sec:proof_main} for how the finite asymptotic growth condition is used.
\end{remark}

\subsection{Compatibility condition}
Before proceeding, we note that the inputs $(c,p,b_Y)$ cannot be entirely independently specified. Indeed, they need to satisfy a compatibility condition so that the class $\Pcal$ is nonempty. We now formally derive this condition. To begin with, we define the quantities
\begin{align} 
    \ell_X(z) & = \frac{1}{2}\Big(\!(c^{-1}\mathrm{div} \, c)_X(z) + \nabla_x \log p(z) + (c_X)^{-1}(z)c_{XY}(z)\big((c^{-1}\mathrm{div} \, c)_Y(z) + \nabla_y \log p(z)\big)\!\Big), \label{eqn:ell_X}\\
    \ell_Y(z) & = \frac{1}{2}\Big(\!(c^{-1}\mathrm{div} \, c)_Y(z) + \nabla_y \log p(z) + (c_Y)^{-1}(z)c_{YX}(z)\big((c^{-1}\mathrm{div} \, c)_X(z) + \nabla_x \log p(z)\big)\!\Big)\! - b_Y(z), \label{eqn:ell_Y} 
\end{align}
which play an important role here and in the sequel. Here, $\mathrm{div} \, c(z)$ is a vector obtained by computing the row-wise divergence of $c$; that is  $\mathrm{div} \, c_i(z) = \sum_{j=1}^{d+m} \partial_j c_{ij}(z)$. Next, assume that a measure $\P \in \Pcal$ is given and let $\phi \in C_c^\infty(D)$ be arbitrary. Then using It\^o's formula we have that
\[\phi(Y_T) = \phi(Y_0) +  M_T + \int_0^T \big(\nabla \phi(Y_t)^\top c_Y(Z_t) b_Y(Z_t)  + \frac{1}{2}\sum_{i,j=1}^m \partial_{ij}\phi(Y_t)(c_Y)_{ij}(Z_t)\big)dt , \]
where 
    $M_T = \int_0^T\sum_{i=1}^m\sum_{j=1}^{d+m} \partial_i \phi(Y_t)c^{1/2}_{d+i,j}(Z_t)dW_{j,t}$ is a local martingale.
We now divide both sides by $T$ and send $T \to \infty$ to deduce that
\[0 = \int_E\int_D \big(\nabla \phi(y)^\top c_Y(x,y)b_Y(x,y) + \frac{1}{2}\mathrm{Tr}(\nabla^2 \phi(y)c_Y(x,y))\big)p(x,y)dydx. \]
Here we used the fact that $\phi$ is bounded, $M_T/T \to 0$ almost surely (see e.g.\ \cite[Lemma~1.3.2]{fernholz2002stochastic}) and the ergodic property \eqref{eqn:ergodic}. Integrating by parts the first term once and the second term twice yields
\[0 = \int_D \phi(y)\int_E \mathrm{div}_y(c_Y(x,y)\ell_Y(x,y)p(x,y))dx dy,\]
where we also switched the order of integration. Since $\phi$ was arbitrary, we have derived, by density of $C_c^\infty(D)$ in $L^1_{\mathrm{loc}}(D)$, the following compatibility condition connecting the inputs $(c,p,b_Y)$,
\begin{equation} \label{eqn:compatibility_condition}
    \int_E \mathrm{div}_y(c_Y(x,y)\ell_Y(x,y)p(x,y))dx = 0 \qquad \text{for  a.e. } y \in D.
\end{equation}
If one fixes the inputs $c$ and $p$ then one explicit way to ensure that $\eqref{eqn:compatibility_condition}$ holds it to  set
\begin{equation} \label{eqn:explicit_compatability}
    b_Y(z) = \frac{1}{2}\Big((c^{-1}\mathrm{div} \, c)_Y(z) + \nabla_y \log p(z) +(c_Y)^{-1}(z)c_{YX}(z)\big((c^{-1}\mathrm{div}\, c)_X(z) + \nabla_x \log p(z)\big)\Big),
\end{equation}
in which case $\ell_Y(z) = 0$ for all $z \in F$. In the examples of Section~\ref{sec:examples} we will specify inputs $(c,p,b_Y)$ that satisfy the compatibility condition \eqref{eqn:compatibility_condition} without imposing $\ell_Y = 0$.

\section{Solving the robust problem: a heuristic approach} \label{sec:heuristic}
The problem \eqref{eqn:lambda_P} is not easily amenable to standard tools from stochastic control. The main difficulty stems from the infinite horizon together with the ergodic constraint \eqref{eqn:ergodic}. Mathematically, the ergodic condition can be thought of as a constraint on the limit as $t \to \infty$ of the marginal distribution of $Z_t$, but it does not offer a more direct restriction of $b_{t,X}^{\P}$. In particular, as it is an infinite-horizon constraint and, so is the criterion \eqref{eqn:growth_rate}, approaches using possibly existing dominated measures for the class $\Pcal$ may be technically challenging.

Instead, we follow a similar approach to \cite{kardaras2021ergodic,itkin2025ergodic} by first restricting our analysis to a suitable class $\Theta_0 \subset \Theta$ of trading strategies that achieve the same growth rate under each admissible measure. Afterwards, we will establish that the robust  optimal strategy over this smaller class is actually globally robust growth-optimal. Here, we take the class
\[\Theta_0  =\{\theta \in \Theta: \theta_t = \nabla_x \phi(X_t,Y_t) \text{ for some } \phi \in C^2_c(F)\}. \]
This is a natural extension of the class
$\{\theta \in \Theta: \theta_t = \nabla \phi(X_t) \text{ for some } \phi \in C^2_c(E)\}$  of functionally generated portfolios considered in \cite{kardaras2021ergodic,itkin2025ergodic} to allow dependence on $Y$, for which more information is available in this setting due to the input $b_Y$. In essence, the portfolios making up $\Theta_0$ can be thought of as stochastic factor dependent functionally generated portfolios and, in this setting, they are needed to obtain robust optimality. From \eqref{eqn:wealth} and It\^o's formula applied to $\phi(Z_t)$ it follows that the logarithmic wealth when using a strategy $\theta^\phi_t := \nabla_x \phi(Z_t) \in \Theta_0$ is
\begin{align}
    \log V_T^{\theta^\phi} & =  \int_0^T \nabla_x \phi(Z_t)^\top dX_t - \frac{1}{2}\int_0^T \nabla_x \phi(Z_t)^\top d[X]_t \nabla_x \phi(Z_t) \\
    & = \phi(Z_T) - \phi(Z_0) - \int_0^T \nabla_y \phi(Z_t)^\top dY_t - \frac{1}{2}\sum_{i,j=1}^{d+m}\int_0^T  \partial_{ij} \phi(Z_t) d[Z_i,Z_j]_t \\
    & \qquad \qquad \qquad \qquad \quad - \frac{1}{2}\int_0^T \nabla_x \phi(Z_t)^\top d[X]_t \nabla_x \phi(Z_t).
\end{align}
Using \ref{defn:class}\ref{item:dynamics} we obtain for every $\P \in \Pcal$ the relationship
\begin{align} \label{eqn:heuristic_wealth1} \log & V_T^{\theta^\phi}  =  \phi(Z_T) - \phi(Z_0) - M_T  \\
& \label{eqn:heuristic_wealth2}  - \int_0^T \big(\nabla_y \phi(Z_t)^\top c_Y(Z_t) b_Y(Z_t) + \frac{1}{2}\mathrm{Tr}(\nabla^2 \phi (Z_t)c(Z_t)) + \frac{1}{2}\nabla_x \phi(Z_t)^\top c_X(Z_t)\nabla_x \phi(Z_t)\big)dt, 
\end{align}
where $M_T = \int_0^T \sum_{i=1}^m\sum_{j=1}^{d+m} \partial_i \phi(Z_t)c_{d+i,j}^{1/2}(Z_t)dW_{j,t}$ is a local martingale. Dividing both sides by $T$ and sending $T \to \infty$, we see by boundedness of $\phi$, and \cite[Lemma~1.3.2]{fernholz2002stochastic} for the local martingale part, that the terms on the right hand side of \eqref{eqn:heuristic_wealth1} vanish. By the ergodic property of Definition~\ref{defn:class}\ref{item:ergodic} the terms in \eqref{eqn:heuristic_wealth2} converge leading to
\[g(\theta^\phi;\P) = -\int_F\big(\nabla_y \phi(z)^\top c_Y(z) b_Y(z) + \frac{1}{2}\mathrm{Tr}(\nabla^2 \phi (z)c(z)) + \frac{1}{2}\nabla_x \phi(z)^\top c_X(z)\nabla_x \phi(z)\big)p(z)dz. \]
Importantly, the right hand side depends on the measure $\P$ \emph{only via $(c,p,b_Y)$}. To make further progress we integrate by parts the second derivative terms. Separately collecting all of the terms involving $\nabla_x \phi$ and $\nabla_y \phi$ allows us to rewrite the asymptotic growth rate as 
\begin{equation}
\begin{aligned}
    g(\theta^\phi;\P) & = \int_F \nabla_x \phi(z)^\top c_X(z) \ell_X(z)p(z)dz - \frac{1}{2}\int_F \nabla_x \phi(z)^\top c_X(z)\nabla_x \phi(z)p(z)dz  \\
    & \qquad + \int_F \nabla_y \phi(z)^\top c_Y(z)\ell_Y(z)p(z)dz, \label{eqn:growth_rate_IBP} 
\end{aligned}
\end{equation}
where we recall that $\ell_X,\ell_Y$ are given by \eqref{eqn:ell_X} and \eqref{eqn:ell_Y}, respectively.

We now seek to put the expression \eqref{eqn:growth_rate_IBP} into a more regular form involving only $\nabla_x \phi$, but not $\nabla_y \phi$ or $\phi$ itself. To accomplish this we further integrate by parts the final term in \eqref{eqn:growth_rate_IBP} to obtain
\begin{equation} \label{eqn:IBP_y}
    \int_F \nabla_y \phi(z)^\top c_Y(z)\ell_Y(z)p(z)dz = -\int_F \phi(z)\mathrm{div}_y(c_Y(z)\ell_Y(z) p(z))dz.
\end{equation}
 Next, we seek to construct a vector field $\mathbf{u}:F \to \R^d$ satisfying
\begin{equation} \label{eqn:div_PDE}
    \mathrm{div}_x \mathbf{u}(\cdot,y) = \mathrm{div}_y(c_Y(\cdot,y)\ell_Y(\cdot,y) p(\cdot,y)), \qquad \text{in } E \text{ for a.e. } y \in D
\end{equation}
together with certain integrability bounds precisely stated in Assumption~\ref{ass:conditions}\ref{item:divergence} of the following section. Once such a $\mathbf{u}$ is found we substitute into \eqref{eqn:IBP_y} and integrate by parts again to obtain
\[\int_F \nabla_y \phi(z)^\top c_Y(z)\ell_Y(z)p(z)dz = \int_F \nabla_x \phi(z)^\top \mathbf{u}(z)dz.\]
Substituting into \eqref{eqn:growth_rate_IBP} gives
\begin{align}
    g(\theta^\phi;\P)&  = \int_F (\xi(z)^\top c_X(z)\nabla_x \phi(z)- \frac{1}{2}\nabla_x \phi(z)^\top c_X(z)\nabla_x \phi(z))p(z)dz \\
    & = \frac{1}{2}\int_F \xi(z)^\top c_X(z)\xi(z)p(z)dz - \frac{1}{2}\int_F (\nabla_x \phi(z) - \xi(z))^\top c_X(z)(\nabla_x \phi(z) - \xi(z))p(z)dz, 
\end{align}
where 
\begin{equation} \label{eqn:xi}
\xi(z) = \ell_X(z) + (c_X)^{-1}(z)\mathbf{u}(z)p^{-1}(z).
\end{equation}
Maximizing the  robust asymptotic growth rate over $\Theta_0$ now amounts to minimizing the functional
\[\phi \mapsto \frac{1}{2}\int_F (\nabla_x \phi(z) - \xi(z))^\top c_X(z)(\nabla_x \phi(z) - \xi(z))p(z)dz \]
over a suitable function space. Under certain assumptions on the inputs (see Assumption~\ref{ass:conditions}) a sufficiently regular solution $\phi^*$ can be found and it satisfies the associated Euler--Lagrange equation 
\begin{align}
    \label{eqn:Euler-Lagrange_xi} 
 & \mathrm{div}_x(c_X(\cdot,y)(\nabla_x \phi^*(\cdot,y) - \xi(\cdot,y))p(\cdot,y))  =0 &   \text{in $E$ for a.e.\ $y \in D$}. \end{align}
This now leads us to a lower bound for $\lambda_\Pcal$ since
\begin{equation} \label{eqn:lower_bound}
\begin{aligned}
\lambda_\Pcal \geq \sup_{\theta^\phi \in \Theta_0}\inf_{\P \in \Pcal} g(\theta^\phi;\P) & = \frac{1}{2}\int_F \xi(z)^\top c_X(z)\xi(z)p(z)dz \\
& \qquad - \frac{1}{2}\int_F (\nabla_x \phi^*(z) - \xi(z))^\top c_X(z)(\nabla_x \phi^*(z) - \xi(z))p(z)dz \\
    & = \frac{1}{2}\int_F\nabla_x \phi^*(z)^\top c_X(z)\nabla_x \phi^*(z) p(z)dz, 
\end{aligned}
\end{equation}
where to obtain the final equality we expanded the quadratic form and used the Euler--Lagrange equation \eqref{eqn:Euler-Lagrange_xi} to formally rewrite the cross term 
\begin{align}
    \int_F \nabla_x \phi^*(z)^\top c_X(z) &\xi(z)p(z)dz  = -\int_F \phi^*(z)\mathrm{div}_x(c_X(z)\xi(z) p(z))dz \\
    & = -\int_F \phi^*(z)\mathrm{div}_x(c_X(z)\nabla_x \phi^*(z) p(z))dz  = \int_F \nabla_x \phi^*(z)^\top c_X(z)\nabla_x \phi^*(z)p(z)dz.
\end{align}
To close the gap and obtain the same upper bound we will construct a worst-case measure $\P^* \in \Pcal$ under which $\theta^* := \theta^{\phi^*}$ is growth-optimal. The key observation, first noted in \cite{kardaras2021ergodic} for their problem but continues to hold in this more general setting, is that the stationary Fokker--Planck equation corresponding to the stochastic differential equation (SDE) as in \eqref{eqn:Z_dynamics} with $ d A_{X,t}^{\P^*} = b_{X,t}^{\P^*} dt = c_X(Z_t)\nabla_x \phi^*(Z_t) dt$ is precisely the Euler--Lagrange equation \eqref{eqn:Euler-Lagrange_xi}. This suggests, formally, that $p$ is the invariant density for $Z$ and, hence, that $\P^* \in \Pcal$. Moreover, under $\P^*$, it is easily verified from the general theory of growth-optimal portfolios (see \cite[Theorem~2.31]{karatzas2021portfolio}) that the growth-optimal strategy is given by $\theta^*$ and its asymptotic growth rate under $\P^*$ is equal to the right hand side of \eqref{eqn:lower_bound}. It follows that
\begin{equation} \label{eqn:upper_bound}
\lambda_\Pcal \leq \sup_{\theta \in \Theta} g(\theta;\P^*) = g( \theta^*;\P^*) = \frac{1}{2}\int_F \nabla_x \phi^*(z)c_X(z)\nabla_x\phi^*(z)p(z)dz.
\end{equation}
The bounds \eqref{eqn:lower_bound} and \eqref{eqn:upper_bound} imply that $\lambda_\Pcal = \frac{1}{2}\int_F \nabla_x \phi^*(z)^\top c_X(z)\nabla_x \phi^*(z)p(z)dz$ and suggest that the robust growth-optimal portfolio is given by $\theta^*=\nabla_x \phi^*(Z)$.

\begin{remark}
Notice that our robust optimal growth problem under ergodicity, as well as the problems in \cite{kardaras2021ergodic} and \cite{itkin2025ergodic}, allow for a finite horizon formulation, when the market is assumed to start in the invariant law and the optimization is averaging over this initial measure. Of course the setup has to be formulated slightly differently to guarantee preservation of the respective invariant measure, but the results remain valid.
\end{remark}

\section{Rigorous problem formulation} \label{sec:rigorous}
\subsection{Assumptions}
The discussion in the previous section was heuristic and requires additional assumptions on the inputs to make the argument rigorous. We now state these assumptions. 

\begin{assum} \label{ass:conditions}
    Let inputs $(c,p,b_Y)$ satisfying Assumption~\ref{ass:inputs} be given and let $\ell_X,\ell_Y$ be as defined in \eqref{eqn:ell_X} and \eqref{eqn:ell_Y}. Assume additionally that the following hold:
    \begin{enumerate}
        \item $\int_F \ell_X(z)^\top c_X(z)\ell_X(z)p(z)dz + \int_F \ell_Y(z)^\top c_Y(z)\ell_Y(z)p(z)dz < \infty$, \label{item:finite_growth}
        \item There exists a measurable  ${\bf u}:F\to \R^d$ which satisfies $\int_F \mathbf{u}(z)^\top (c_X)^{-1}(z)\mathbf{u}(z)p^{-1}(z)dz < \infty$ and is a weak solution to \eqref{eqn:div_PDE}. That is,
        \[-\int_E \mathbf{u}(x,y)\nabla \psi(x)dx = \int_E \mathrm{div}_y(c_Y(x,y)\ell_Y(x,y)p(x,y))\psi(x)dx\]
     for all $\psi \in C_c^\infty(E)$ and for a.e.\ $y \in D$.
 \label{item:divergence} 
        \item \label{item:test_functions}
        There exist functions $\chi_n \in C_c^\infty(F)$ satisfying $0 \leq \chi_n \leq 1$, $\lim_{n \to \infty} \chi_n = 1$ and 
        \[\lim_{n \to \infty} \int_F \nabla_x\chi_n(z)^\top c_X(z)\nabla_x\chi_n(z)p(z)dz  = \lim_{n \to \infty} \int_F \nabla_y \chi_n(z)^\top c_Y(z)\nabla_y \chi_n(z)p(z)dz = 0. \] 
    \end{enumerate}
\end{assum}
\begin{remark} \label{rem:u}
    We remark that the compatibility condition \eqref{eqn:compatibility_condition} is not \emph{explicitly} assumed. However, we cannot expect Assumption~\ref{ass:conditions}\ref{item:divergence} to hold without it. This is most easily seen in the case that $d=1$, $E = \R$ and $c_X$ is bounded. In this case the unique (up to additive function of $y$) solution to \eqref{eqn:div_PDE} is \begin{equation} \label{eqn:u_d=1}
        \mathbf{u}(x,y) = \int_{-\infty}^x \mathrm{div}_y(c_Y(x',y)\ell_Y(x',y)p(x',y))dx'.
    \end{equation}Since $p$ is a density on $\R$ we must have $\lim_{|x|\to \infty} p(x) = 0$ so that, in particular, $p^{-1}(x) \to \infty$ as $|x| \to \infty$. Consequently, since $c_X$ is bounded, a necessary condition for $\int_F \mathbf{u}^2(z)(c_X)^{-1}(z)p^{-1}(z)dz$ to be finite is for $\lim_{|x| \to \infty} \mathbf{u}(x,y) = 0$ for a.e.\ $y \in D$. Since the lower bound of integration for the integral defining $\mathbf{u}$ is $-\infty$, we have $\lim_{x \to -\infty} \mathbf{u}(x,y) = 0$, but the remaining condition amounts to
\[0 = \lim_{x \to \infty} \mathbf{u}(x,y) = \int_{-\infty}^\infty \mathrm{div}_y(c_Y(x',y)\ell_Y(x',y)p(x',y))dx', \quad \text{ for a.e.\ } y \in D,\]
which is precisely \eqref{eqn:compatibility_condition}.
\end{remark}
Assumption~\ref{ass:conditions}\ref{item:finite_growth} is needed to ensure finiteness of the robust growth rate $\lambda_\Pcal$. The second assumption is stated in a fairly abstract form but can be reduced to certain explicitly checkable integrability bounds. Indeed, if $d = 1$ then \eqref{eqn:div_PDE} becomes an ODE with unique (up to additive constant) solution given as in \eqref{eqn:u_d=1} with lower bound of integration given by $\inf\{x \in E\}$. For $d \geq 2$, we can define the Newtonian potential 
\[\Phi(x) = \begin{cases}
    \frac{1}{2\pi}\log|x|, & d = 2, \\
    \frac{1}{d(d-2)\alpha(d)}\frac{1}{|x|^{d-2}}, & d \geq 3
\end{cases},\]
where $\alpha(d)$ is the Lebesgue measure of a unit ball in $\R^d$, set \[v(x,y) = \int_E \Phi(x-x')\mathrm{div}_y(c_Y(x',y)\ell_Y(x',y)p(x',y))dx'\] and note that (at least formally) $\mathbf{u} = \nabla_x v$ is a weak solution to \eqref{eqn:div_PDE}. Indeed, by the theory of Poisson equations we have $\mathrm{div}_y(c_Y\ell_Yp) = \Delta_x v = \mathrm{div}_x \mathbf{u}$ in the weak sense, where we also used that $\Delta_x = \mathrm{div}_x \circ \nabla_x$ is the Laplacian. As such, irrespective of the dimension, only the integrability condition $\int_F \mathbf{u}(z)^\top (c_X)^{-1}(z)\mathbf{u}(z)p^{-1}(z)dz < \infty$ needs to be checked for the explicitly constructed $\mathbf{u}$ to ensure that Assumption~\ref{ass:conditions}\ref{item:divergence} holds.  Since the theory of divergence equations of the type \eqref{eqn:div_PDE} is rich and solutions other than the one constructed above may be desirable, we pose the assumption in the stated form. We explicitly construct a ${\bf u}$ solving \eqref{eqn:div_PDE} and the required integrability condition for our multivariate Ornstein--Uhlenbeck example of Section~\ref{sec:OU}. Assumption~\ref{ass:conditions}\ref{item:test_functions} is standard for these problems and is equivalent to the existence of a recurrent symmetric Markov process with covariance matrix $c$ and invariant density $p$ (see \cite[Theorem~1.6.3]{fukushimadirichlet994}). Analogously to \cite{kardaras2021ergodic,itkin2025ergodic} this assumption is crucial to ensure the class $\Pcal$ is nonempty.
\subsection{The class \texorpdfstring{$\Pcal_0$}{P₀}}

To state our main results in the next section we first need to introduce a sub class of $\Pcal$. 
\begin{defn}[Finite growth class] \label{defn:P0}
    We define the class $\Pcal_0$ to be all measures $\P \in \Pcal$ which additionally satisfy
    \begin{enumerate}
    \setcounter{enumi}{2}
        \item  $\sup_{\theta \in \Theta}g(\theta;\P) < \infty.$ \label{item:Pcal0}
    \end{enumerate}
\end{defn}
The class $\Pcal_0$ restricts to measures under which infinite growth is impossible to achieve. By the general theory of growth-optimization (see Chapter~2 of \cite{karatzas2021portfolio} and, in particular, their Theorem~2.31), this amounts to a requirement that the finite variation part of $X$ be in the range of the quadratic variation process. Namely, any measure $\P \in \Pcal_0$ admits a drift of the form $dA^\P_{X,t} = c_X(Z_t)b^\P_{X,t}dt$ for some progressively measurable process $b^\P_X$. Moreover, it is easy to verify that the growth-optimal portfolio under $\P$ is given precisely by $ \theta^*_\P = b^\P_X$. A direct computation shows that its logarithmic wealth under $\P$ is given by
\[\log V^{\theta^*_\P}_T = \frac{1}{2}[L^\P]_T + L^\P_T,\qquad \text{where} \qquad L^\P_T = \int_0^T (b^\P_{X,t})^\top c_X^{1/2}(Z_t)dW_t\] is a local martingale. On the set $\{[L^\P]_\infty < \infty\}$ it follows from \cite[Lemma~1.3.2]{fernholz2002stochastic} that $L^\P_T/T \to 0$, $\P$-a.s.\ as $T \to \infty$. In this case $g( \theta^*_\P;\P) = 0$ so that \ref{item:Pcal0} is satisfied. Conversely, on the set $\{[L^\P]_\infty = \infty \}$, we can write
\[\frac{1}{T}\log V_T^{ \theta^*_\P} = \frac{[L^\P]_T}{T}\left(\frac{1}{2} + \frac{L^\P_T}{[L^\P]_T}\right).\] The Dambis, Dubins-Schwarz Theorem (\cite[Theorem~V.1.6]{revuz1999continuous}) together with the strong law of large numbers for Brownian motion ensures that $\lim_{T \to \infty} L_T^\P/[L^\P]_T = 0$, $\P$-a.s. Hence, the finite growth condition \ref{item:Pcal0} amounts to the requirement
\begin{equation} \label{eqn:bP_condition}
    \P\left(\liminf_{T \to \infty} \frac{1}{T}\int_0^T (b^\P_{X,t})^\top c_X(Z_t)b^\P_{X,t}dt < \infty\right) > 0,
\end{equation}
where we used the fact that $[L^\P]_T = \int_0^T (b^\P_{X,t})^\top c_X(Z_t)b^\P_{X,t}dt$.

\section{Main results} \label{sec:main_results}

We are now ready to state our main results. The proofs are all postponed to Appendix~\ref{sec:proofs} for better readability. We start with a lemma guaranteeing the existence of, what turns out to be, the feedback form function characterizing the optimal strategy.
\begin{lem}[Characterization of the optimizer] \label{lem:variational}
Set 
    \[\Dcal = \{\phi: F \to \R: \phi \text{ is measurable and } \phi(\cdot,y) \in C^2(E) \text{ for a.e.\ } y \in D\}.\]
    Let Assumption~\ref{ass:inputs} and Assumption~\ref{ass:conditions}\ref{item:finite_growth}-\ref{item:divergence} be satisfied. Then there exists $\phi^* \in \Dcal$ satisfying
\begin{equation} \label{eqn:variational}
\phi^*  \in \arg\min_{\phi \in \Dcal} \int_F(\nabla_x \phi(z) - \xi(z))^\top c_X(z)(\nabla_x \phi(z) - \xi(z))p(z)dz,
\end{equation}
where we recall that $\xi$ is given by \eqref{eqn:xi}.
Moreover, $\phi^*$ is unique up to an additive function of $y$, $\nabla_x \phi^* \in L^q_{\mathrm{loc}}(F;\R^d)$ for every $q \in [2,\infty)$ and $\phi^*$ satisfies the Euler--Lagrange PDE 
 \begin{align}
 \label{eqn:Euler-Lagrange}
    & \mathrm{div}_x\big(c_X(\cdot,y)(\nabla_x \phi^*(\cdot,y) - \ell_X(\cdot,y)) p(\cdot,y)\big)  = \mathrm{div}_y(c_Y(\cdot,y)\ell_Y(\cdot,y)p(\cdot,y))
\end{align}
in $E$ for a.e.\ $y \in D$.
\end{lem}

We now state the main result of this paper.
\begin{thm}[Main result] \label{thm:main}  The robust asymptotic growth rate satisfies
\begin{equation} \label{eqn:lambda_P_thm}
    \lambda_\Pcal  = \frac{1}{2}\int_F \nabla_x \phi^*(z)^\top c_X(z)\nabla_x \phi^*(z)p(z)dz, 
\end{equation}
where $\phi^*$ is as in Lemma~\ref{lem:variational}. 
Moreover, the strategy $\theta^*_t = \nabla_x \phi^*(Z_t)$ is robust growth-optimal in the sense that $\lambda_{\Pcal} = g(\theta^*;\P)$ for every $\P \in \Pcal_0$.    
\end{thm}


An important part of the proof, as well as a key component of our numerical examples in Section~\ref{sec:examples}, involves establishing and characterizing the worst-case measure.
\begin{prop}[Worst-case measure] \label{prop:worst_case} Let $z \in F$ be arbitrary.
    The SDE 
        \begin{equation} \label{eqn:worst_case_SDE} 
        dZ_t = \begin{pmatrix}
        c_X(Z_t)\nabla_x\phi^*(Z_t) \\ c_Y(Z_t)b_Y(Z_t)               
        \end{pmatrix}dt + c^{1/2}(Z_t)dW_t, \quad Z_0 = z,
        \end{equation}
where $\phi^*$ is given by Lemma~\ref{lem:variational}, admits a weak solution $\P^*_z \in \Pcal_0$. 
\end{prop}

The measure $\P^*_z$ is called a worst-case measure because $\theta^*$ is growth-optimal under $\P^*_z$ so that the maximal achievable asymptotic growth rate under $\P^*_z$ is the robust growth rate $\lambda_\Pcal$. 
In the sequel, the initial condition $z \in F$ will not play an impactful role and, as such, we will frequently omit it from the notation, referring by $\P^*$ to the law of the process $Z$ with dynamics \eqref{eqn:worst_case_SDE} for any arbitrary initial value $z \in F$.

\section{Financial insights} \label{sec:financial_insights}
In this section we discuss the financial insights of the results in Section~\ref{sec:main_results}. In particular, we compare the robust optimal strategy $\theta^*$, worst-case measure $\P^*$ and robust optimal growth rate $\lambda_{\Pcal}$ obtained here to their counterparts $\widehat \theta$, $\widehat \P$ and $\lambda_\Pi$ from \cite{itkin2025ergodic}. We additionally discuss the dependence of $\theta^*$ on $Y$ and potential computational savings for solving the PDE \eqref{eqn:Euler-Lagrange}. Section~\ref{sec:examples}, which follows, supports and expands on the general financial insights in specific examples of interest.

\subsection{Summary of results in \texorpdfstring{\cite{itkin2025ergodic}}{Itkin et.\ al.\ (2025)}} \label{sec:recap_other}
To facilitate the financial discussion we briefly summarize the key findings of \cite{itkin2025ergodic} and, in the process, establish notation necessary for the subsequent sections. We focus here on the key results and main conclusions obtained from \cite[Section~5.2]{itkin2025ergodic} and refer the reader to the full article for details regarding the precise technical conditions under which the results there were proven.

The setup in \cite[Section~5.2]{itkin2025ergodic} took only the pair $(c,p)$, encompassing covariation matrix and invariant density functions respectively, as inputs and did therefore not restrict the drift of $Y$. As such, the focus there is on the class of probability measures $\Pi$ on $(\Omega,\F)$ consisting of all measures $\P$ under which\footnote{A third requirement that the laws of $(Z_t)_{t \geq 0}$ are tight under $\P$ was also required in \cite{itkin2025ergodic}, but can be removed by restricting the class of measures to have finite asymptotic growth as in the definition of $\Pcal_0$ (see also the proof of Theorem~\ref{thm:main} in Appendix~\ref{sec:proof_main}).}
\begin{enumerate}[noitemsep]
    \item \label{item:QV} $Z$ is a continuous semimartingale with quadratic variation process $ 
    [Z]_T  = \int_0^Tc(Z_t)dt$ and
    \item $Z$ satisfies the ergodic property \eqref{eqn:ergodic}. 
\end{enumerate}
The class of measures which additionally prohibits infinite asymptotic growth, as in item \ref{item:Pcal0} of Definition~\ref{defn:P0}, will be denoted by $\Pi_0$. 
The corresponding problem is then to characterize
\begin{equation}
    \lambda_{\Pi} = \sup_{\theta \in \Theta}\inf_{\P \in \Pi} g(\theta;\P)
\end{equation}
and to find the robust growth-optimal strategy $\widehat \theta$ (the growth rate $g_{\mathrm{prob}}(\cdot,\cdot)$ of \eqref{eqn:growth_rate_prob} was used in \cite{itkin2025ergodic}, but the results can be extended to the growth rate $g(\cdot,\cdot)$ of \eqref{eqn:growth_rate} as discussed in Remark~\ref{rem:growth_rate}). The key conclusion of \cite{itkin2025ergodic} was that the growth-optimal strategy $\widehat \theta$ is functionally generated by a function of $x$ only and, consequently, does not depend on the stochastic factor $Y$. Indeed, $\widehat \theta_t = \nabla \widehat \phi(X_t)$, where $\widehat \phi \in C^2(E)$ is the unique (up to additive constant) solution to the Euler--Lagrange equation
\begin{equation} \label{eqn:Euler_Lagrange_x}
\mathrm{div}\Big(a(x)\nabla \widehat \phi(x) - \frac{1}{2}\mathrm{div}\big(a(x)\big)\Big) = 0,
\end{equation}
where
\begin{equation} \label{eqn:a} 
a_{ij}(x) = \int_D (c_X)_{ij}(x,y)p(x,y)dy, \qquad i,j=1,\dots,d, \quad x \in E
\end{equation}
is the marginal volatility matrix of $X$ obtained by integrating out the variable $y$ with respect to the ergodic density $p$.
The corresponding robust growth rate is then given by
\begin{equation} \label{eqn:lambda_Pi}
\lambda_{\Pi} = \frac{1}{2}\int_{D\times E} \nabla\widehat \phi(x)^\top c_X(x,y)\nabla \widehat \phi(x)p(x,y)dydx = \frac{1}{2}\int_E \nabla \widehat \phi(x)^\top a(x) \nabla \widehat \phi(x)dx
\end{equation}
and $\widehat \theta$ is robust growth-optimal in the sense that $g(\widehat \theta;\P) = \lambda_{\Pi}$ for all $\P \in \Pi_0$. 

The corresponding worst-case measure $\widehat \P$, under which $\widehat \theta$ is growth-optimal, has to satisfy $b^{\widehat \P}_{X,t} = c_X(Z_t)\nabla \widehat \phi(X_t)$, but, unlike the setup of this paper, this does not yet specify the dynamics of $Z$. Indeed, the drift of $Y$, which is not constrained in \cite{itkin2025ergodic}, must be carefully chosen so as to yield an admissible worst-case measure. The difficulty is to find a drift process $b^{\widehat \P}_Y$ under which the process $Z$ is nonexplosive (i.e.\ stays in the domain $F$) and has $p$ as its long-term invariant density. In \cite{itkin2025ergodic} it was shown that setting $b^{\widehat \P}_{Y,t} = \nabla_y \widehat v(Z_t)$, where $\widehat v$ satisfies $\int_F \nabla_y \widehat v(z)^\top c_Y(z)\nabla_y \widehat v(z)p(z)dz < \infty$ and the PDE \cite[Equation (5.20)]{itkin2025ergodic} ensures that $\widehat \P \in \Pcal_0$. The PDE for $\widehat v$, in the notation of this paper, is given by
\begin{equation} \label{eqn:v_PDE}
    \mathrm{div}_y(c_Y(z)(\ell_Y^0(z)-\nabla_y \widehat v(z))p(z)) = -\mathrm{div}_x (c_X(z)(\ell_X(z)-\nabla \widehat \phi(x))p(z)), \qquad z \in F 
\end{equation}
where 
\begin{align}
    \ell_Y^0(z) 
     = \frac{1}{2}\Big((c^{-1}\mathrm{div} \, c)_Y(z) + \nabla_y \log p(z) + (c_Y)^{-1}(z)c_{YX}(z)\big((c^{-1}\mathrm{div} \, c)_X(z) + \nabla_x \log p(z)\big)\Big) \label{eqn:ly0}
\end{align}
is the part of $\ell_Y$ that is not affected by the input $b_Y$. As such, the dynamics of $Z$ under the worst-case measure $\widehat \P$ are given by
\begin{equation} \label{eqn:hatP_dynamics}
d\begin{pmatrix}
    X_t \\ Y_t
\end{pmatrix} = \begin{pmatrix}
    c_X(Z_t)\nabla \widehat \phi(X_t) \\
    c_Y(Z_t)\nabla_y \widehat v(Z_t)
\end{pmatrix}dt + c^{1/2}(Z_t)dW_t.    
\end{equation}
The PDE \eqref{eqn:v_PDE} formally corresponds to the stationary Fokker--Planck equation associated with \eqref{eqn:hatP_dynamics} and in \cite[Lemma~5.11]{itkin2025ergodic} it was shown, under appropriate technical assumptions, that a solution $\widehat v$ to \eqref{eqn:hatP_dynamics} exists and satisfies the required integrability bounds. In Section~\ref{sec:OU} below we will explicitly compute the worst-case measure $\widehat \P$ by solving the PDE \eqref{eqn:v_PDE} in the case where the diffusion matrix $c$ is constant and $p$ is a Gaussian density. 

\subsection{Discussion} \label{sec:discussion}
Clearly $\Pcal \subset \Pi$ from which it immediately follows that $\lambda_\Pcal \geq \lambda_\Pi$. In fact, there is typically a strict inequality and the gap in robust growth between the setups can be quantified. Since $\widehat \theta$ achieves the same growth rate for every measure in $\Pi_0$ and $\P^* \in \Pcal_0 \subset \Pi_0$ we can deduce that
\begin{align}
    \lambda_\Pi & = g(\widehat \theta;\P^*) = \lim_{T \to \infty} \frac{1}{T}\log V^{\widehat \theta}_T \\
    & = \lim_{T \to \infty}\frac{1}{T}\int_0^T (\nabla \widehat \phi(X_t)c_X(Z_t)\nabla_x\phi^*(Z_t) - \frac{1}{2}\nabla \widehat \phi(X_t)^\top c_X(Z_t)\nabla \widehat \phi(X_t))dt  + \lim_{T \to \infty}\frac{M_T}{T} \\
    & = \int_F \nabla \widehat \phi(x)^\top c_X(z)\nabla_x \phi^*(z)p(z)dz - \frac{1}{2}\int_F \nabla \widehat \phi(x)^\top c_X(z)\nabla \widehat \phi(x)p(z)dz,
\end{align}
where $dM_t = \sum_{i=1}^d\sum_{j=1}^{d+m}\partial_i \widehat \phi(X_t)(c_X)_{ij}(Z_t)dW_{j,t}$, the limits are understood $\P^*$-a.s.\ and the final equality followed by the ergodic property \eqref{eqn:ergodic} and  the fact that $M_T/T \to 0$ as $T \to \infty$ by \cite[Lemma~1.3.2]{fernholz2002stochastic}. 
As such, from \eqref{eqn:lambda_P_thm} we obtain
\begin{equation} \label{eqn:growth_gap}
    \lambda_{\Pcal} - \lambda_{\Pi} = \frac{1}{2}\int_F (\nabla_x \phi^*(z) - \nabla \widehat \phi(x))^\top c_X(z)(\nabla_x \phi^*(z) - \nabla \widehat \phi(x))p(z)dz.
\end{equation}
From this expression we see that assuming the drift of $Y$ is known leads to an improvement in robust growth rate in comparison to the setup that additionally robustifies over the drift of $Y$. The only exception is when $\nabla_x \phi^* = \nabla \widehat \phi$, in which case $\theta^* = \widehat \theta$ and the conditional law of $X$ given $Y$ coincides under both worst-case measures $\P^*$ and $\widehat \P$.
In Section~\ref{sec:OU}, where we study a tractable high-dimensional example, we show that this edge case can happen, but is atypical and corresponds to a very specific adversarial choice for the input $b_Y$. In the remainder of this section we assume that $\theta^* \ne \widehat \theta$.

From \eqref{eqn:growth_gap} we deduce that the investor can strictly improve their robust growth rate if they know the dynamics of the factor process $Y$. A natural question is the following: what happens to the investor's growth rate if they believe to know the correct dynamics of $Y$, but are mistaken. In this case they would use the strategy $\theta^*$, but the true measure $\P$ driving the dynamics of $Z$ is a member of $\Pi\setminus \Pcal$. As $\P$ is unknown, a natural quantity to consider is the maximum potential loss in growth from using $\theta^*$ compared to the $\Pi$-growth-optimal strategy $\widehat \theta$. Although we are unable to compute this quantity exactly we are able to get a lower bound on it since $\widehat \P \in \Pi \setminus \Pcal$. From this observation we obtain
\begin{align} 
\lambda_{\Pi} - \inf_{\P \in \Pi}g(\theta^*;\P)  \geq \lambda_{\Pi} - g(\theta^*;\widehat \P) & = \frac{1}{2}\int_F (\nabla \widehat \phi(x)-\nabla_x \phi^*(z))^\top c_X(z)(\nabla \widehat \phi(x)-\nabla_x \phi^*(z))p(z)dz \\
& = \lambda_{\Pcal} - \lambda_{\Pi} \label{eqn:growth_rate_differential},
\end{align}
where in the penultimate equality we used \eqref{eqn:lambda_Pi} in place of $\lambda_\Pi$, directly computed $g(\theta^*;\widehat \P)$, akin to how we computed $g(\widehat \theta;\P^*)$ above, and collected like terms. The final equality follows from \eqref{eqn:growth_gap}. This computation shows that the loss in growth the investor can suffer, relative to the benchmark $Y$-independent strategy $\widehat \theta$, if they infer the incorrect dynamics of $Y$ can be \emph{at least} as large as their gains from utilizing the stochastic factor (when they correctly posit the dynamics of $Y$). The upshot is that incorporating the stochastic factor in one's trading strategy can lead to an improvement in robust growth rate but, if the dynamics of the factor are incorrectly specified, can also lead to underperformance that is at least as large. As such, whether the investor should utilize the strategy $\theta^*$ or $\widehat \theta$ depends strongly on their confidence in the dynamics of the stochastic factor process they are estimating. This relationship is borne out in Figure~\ref{fig:OU_boxplot} of Section~\ref{sec:pairs_CTOU}, where we study an application to pairs trading.

We conclude this section with a discussion of how the stochastic factor process $Y$ affects the strategy $\theta^*$. Somewhat surprisingly, the feedback form function $\phi^*$ specifying $\theta^*$ is obtained by solving a collection of PDEs \eqref{eqn:Euler-Lagrange} indexed by the variable $y$. In other words, each state $y \in D$ has its own autonomous PDE that should be solved to determine the investor's optimal holdings when the process $Y$ takes the value $y$. Whence, although the optimal strategy depends on the high-dimensional $(d+m)$-dimensional process $(X,Y)$, the dimension of equation \eqref{eqn:Euler-Lagrange} that needs to be solved is only $d$-dimensional. In the context of modern machine learning applications involving high-dimensional features, one may expect $m \gg d$ in many applications making this dimension reduction significant. 

Moreover, the dissection of the PDE \eqref{eqn:Euler-Lagrange} into $y$-slices has additional benefits, from a computational point of view, as the optimal strategy $\nabla_x \phi^*(\cdot,y)$ can be solved in an online manner the first time that the factor process $Y$ takes the value $y$.\footnote{Perhaps after discretizing and binning the $y$-states.} In particular, one does not need to compute the strategy at values of the large dimensional state space $D$ which are never observed in practice. Additionally, the right hand side of \eqref{eqn:Euler-Lagrange} can be computed offline and, as such, does not increase the online run time to solve the PDE \eqref{eqn:Euler-Lagrange} and obtain the investor's robust growth-optimal holdings.


\section{Examples} \label{sec:examples}
Below we consider several examples illustrating the theoretical results we derived above. In this section we freely use the notation established in the previous parts of the paper and, in particular the notation established in Section~\ref{sec:discussion} for the quantities studied in \cite{itkin2025ergodic}. In Section~\ref{sec:gradient} we describe how an explicit, in terms of the inputs $(c,p,b_Y)$ and ${\bf u}$ of Assumption~\ref{ass:inputs}\ref{item:divergence}, solution to \eqref{eqn:Euler-Lagrange} can be obtained when $\xi$ is a gradient. Section~\ref{sec:OU} contains an in-depth treatment when our setup is compatible with an Ornstein--Uhlenbeck (OU) specification. Assumption~\ref{ass:conditions} is carefully verified and optimal strategies, growth rates and worst-case measures in both the setup of this paper and the previous work \cite{itkin2025ergodic} are explicitly computed in arbitrary dimension. Section~\ref{sec:pairs_trading_framework} then sets $d=m=1$ and explores how our results can be applied to a pairs trading application. Section~\ref{sec:pairs_CTOU} complements the theoretical results of Section~\ref{sec:OU} with numerical simulations in the pairs trading context illustrating the behaviour of the strategies and quantifying their growth rates. The remaining sections, still through the lens of the pairs trading application, illustrate our results in more general specifications. Section~\ref{sec:fat_tails} focuses on an extension when the invariant measure is a bivariate $t$-distribution, capturing the tendency for financial returns to have fat tails, while Section~\ref{sec:stoch_vol} explores what happens when $Y$ is taken to be the stochastic volatility of $X$.\footnote{All of the code used to produce the numerical results of this paper are freely available on \href{https://github.com/balintg1994/stochastic-factors-can-matter-improving-robust-growth-under-ergodicity}{GitHub}. }

\subsection{Gradient case \texorpdfstring{$\xi = \nabla_x h$}{}} \label{sec:gradient}
If there exists a function $h: F \to \R$ such that $\xi =\nabla_x h$, then we directly see from \eqref{eqn:Euler-Lagrange_xi} that $\phi^* = h$ solves the Euler--Lagrange equation. In this case the optimal strategy is given by
\begin{equation} \label{eqn:gradient_case}
    \theta^*_t = \nabla_x h(Z_t) =\xi(Z_t) = \ell_X(Z_t) + (c_X)^{-1}(Z_t)\mathbf{u}(Z_t)p^{-1}(Z_t)
\end{equation}
When $d=1$, the gradient condition always holds regardless of the dimension of $D$ since any integrable univariate function is the derivative of its integral.
\subsection{Ornstein--Uhlenbeck dynamics} \label{sec:OU}

Here we set $E = \R^d$, $D = \R^m$ and we assume that the volatility matrix is constant, the invariant density is a centered Gaussian density and the drift of $y$ is affine. That is we take
\begin{equation} \label{eqn:Gaussian_inputs}
c(z) = c, \qquad p(z) = (2\pi)^{-\frac{d+m}{2}} (\det \Sigma)^{-1/2}\exp(-\frac{1}{2}z^\top \Sigma^{-1}z), \qquad b_Y(z) = \frac{1}{2}(\alpha + \beta z),
\end{equation}
where $c, \Sigma \in \mathbb{S}^{d+m}_{++}$ are the instantaneous and stationary covariance matrices respectively,

and $\alpha \in \R^m$, $\beta \in \R^{m\times (d+m)}$ are arbitrary at the moment, but will be chosen to ensure that the compatibility condition \eqref{eqn:compatibility_condition} is satisfied. In this section, given an invertible matrix $M$ we write $M^{-1}_X$ for $(M^{-1})_X$, and similarly for $M^{-1}_{XY},M^{-1}_{YX},M^{-1}_Y$, as these quantities appear frequently in the ensuing calculations and this convention improves the aesthetics of the lengthy formulas. We continue to write $(M_X)^{-1}$ and $(M_Y)^{-1}$ for the inverse of a particular block. 

Next we compute that
\[\ell_X(x,y) = -\frac{1}{2}(Ax + By), \qquad \ell_Y(x,y) = -\frac{1}{2}(Cx + Dy - \alpha),\]
where
\begin{align}
    A & = \Sigma^{-1}_X + (c_X)^{-1}c_{XY}\Sigma^{-1}_{YX}, \qquad & B  = \Sigma^{-1}_{XY} + (c_X)^{-1}c_{XY}\Sigma^{-1}_Y, \\
    C & = \Sigma^{-1}_{YX}+(c_Y)^{-1}c_{YX}\Sigma^{-1}_X+\beta_X, \qquad & D = \Sigma^{-1}_Y + (c_Y)^{-1}c_{YX}\Sigma^{-1}_{XY} + \beta_Y, 
\end{align}
and we split up $\beta = \begin{bmatrix} \beta_X & \beta_Y \end{bmatrix}$ so that $\beta_X \in \R^{m\times d}$ and $\beta_Y \in \R^{m \times m}$. By interchanging derivative and integral we see that the compatibility condition \eqref{eqn:compatibility_condition} here becomes 
\begin{align}
    0 = \int_{\R^d}\mathrm{div}_y(c_Y\ell_Y(x,y)p(x,y))dx & = -\frac{1}{2} \mathrm{div}_y\left(c_Y\int_{\R^d}(Cx + Dy -\alpha)p(x,y)dx\right) \\
    & = -\frac{1}{2}\mathrm{div}_y\Big(c_Y\big((C\Sigma_{XY}(\Sigma_Y)^{-1}+D)y - \alpha\big)p_Y(y)\Big),
\end{align}
where in the final equality we computed the conditional expectation of $CX|Y$ for $(X,Y) \sim N(0,\Sigma)$ and denoted by $p_Y$ the marginal density of $Y$.
It's clear that this condition holds for every $y \in \R^m$ if and only if $\alpha = 0$ and $C\Sigma_{XY}(\Sigma_Y)^{-1} +D = 0$. The latter condition amounts to requiring
\begin{equation} \label{eqn:beta_Y}
\beta_Y = -(\Sigma^{-1}_Y + (c_Y)^{-1}c_{YX}\Sigma^{-1}_{XY} + C\Sigma_{XY}(\Sigma_Y)^{-1}), 
\end{equation}
while the input $\beta_X$ remains free.

With the inputs now fixed we check Assumption~\ref{ass:conditions}. Item \ref{item:finite_growth} clearly holds as $\ell_X$ and $\ell_Y$ are both linear and item \ref{item:test_functions} is easily seen to hold due to the Gaussian tails of $p$ by choosing, for example, $\chi_n = 1*\eta_{1/n}$ for a standard mollifier $(\eta_\epsilon)_{\epsilon > 0}$. It remains to establish Assumption~\ref{ass:conditions}\ref{item:divergence}. Here, the PDE \eqref{eqn:div_PDE} for ${\bf u}$ is 
 \begin{equation} \label{eqn:u_Gaussian}
 \mathrm{div}_x {\bf u}(x,y) = -\frac{1}{2}\mathrm{div}_y\big(c_YC(x-\Sigma_{XY}(\Sigma_Y)^{-1}y)p(x,y)\big).
 \end{equation}
An explicit solution is given by
\[ {\bf u}(x,y) = - \frac{1}{2}(\Sigma_X^{-1})^{-1}C^\top c_Y(\Sigma_{YX}^{-1}x + \Sigma_Y^{-1}y)p(x,y).\]
Indeed, writing $M = (\Sigma_X^{-1})^{-1}C^\top c_Y$ we have that
\begin{align}
\mathrm{div}_x {\bf u}(x,y) & = \frac{1}{2}\mathrm{div}_x(M\nabla_y p(x,y)) = \frac{1}{2} \sum_{i=1}^d\sum_{j=1}^m M_{ij}\partial_{x_i}\partial_{y_j}p(x,y)=  \frac{1}{2} \sum_{j=1}^m\sum_{i=1}^d M^\top_{ji}\partial_{y_i}\partial_{x_i}p(x,y) \\
& = \frac{1}{2}\mathrm{div}_y(M^\top \nabla_x p(x,y)) = -\frac{1}{2}\mathrm{div}_y(c_YC(\Sigma_X^{-1})^{-1}(\Sigma^{-1}_X x + \Sigma^{-1}_{XY}y)p(x,y)) \label{eqn:u_verification_Gaussian} \\
& = -\frac{1}{2}\mathrm{div}_y\big(c_YC(x-\Sigma_{XY}(\Sigma_Y)^{-1}y)p(x,y)\big).
\end{align}
In the last step we used that $(\Sigma_X^{-1})^{-1}\Sigma^{-1}_{XY} = -\Sigma_{XY}(\Sigma_Y)^{-1}$, which follows from the fact that $\Sigma^{-1}$ is the inverse of $\Sigma$. Indeed, from this inverse relationship we have that $\Sigma^{-1}_X\Sigma_{XY} + \Sigma^{-1}_{XY}\Sigma_Y = 0$, so bringing the first term to the other side of the equality and multiplying both sides on the left by $(\Sigma_X^{-1})^{-1}$ and on the right by $(\Sigma_Y)^{-1}$ establishes the identity. Moreover, we have
\begin{align}
    \int_{\R^{d+m}} {\bf u}^\top(z) & (c_X)^{-1}(z)  {\bf u}(z) p^{-1}(z)dz \\
    & = \frac{1}{4}\int_{\R^{d+m}} (x-\Sigma_{XY}\Sigma_Y^{-1}y)^\top M^\top (c_X)^{-1}M(x-\Sigma_{XY}\Sigma_Y^{-1}y)p(x,y)dxdy < \infty,
\end{align}
so that Assumption~\ref{ass:conditions}\ref{item:divergence} holds. 

With the assumptions verified, Theorem~\ref{thm:main} yields a robust optimal strategy characterized by the solution of the PDE \eqref{eqn:Euler-Lagrange}. Here we observe that both $\ell_X$ and $(c_X)^{-1} {\bf u} p^{-1}$ are linear in $x$, so that $\xi(x,y) = \nabla_x Q(x,y)$ for some quadratic function $Q$. It follows from the discussion in Section~\ref{sec:gradient} that the optimal strategy is given by
\[\theta^*_t = \nabla_x \phi^*(Z_t) =  -\frac{1}{2}\underbrace{(A+(c_X)^{-1}(\Sigma_X^{-1})^{-1}C^\top c_Y\Sigma^{-1}_{YX})}_{M_X}X_t -\frac{1}{2} \underbrace{(B + (c_X)^{-1}(\Sigma_X^{-1})^{-1}C^\top c_Y\Sigma_Y^{-1})}_{M_Y}Y_t\] 
and the robust-optimal growth rate is given by
\begin{align}
    \lambda_{\Pcal} & = \frac{1}{2}\int_{\R^{d+m}}\nabla_x \phi^*(z)^\top c_X\nabla_x \phi^*(z)p(z)dz \\
    & = \frac{1}{8}\mathrm{Tr}(M_X^\top c_X M_X\Sigma_X)+ \frac{1}{4}\mathrm{Tr}(M_X^\top c_XM_Y\Sigma_{YX}) + \frac{1}{8}\mathrm{Tr}(M_Y^\top c_XM_Y\Sigma_Y). \label{eqn:lambda_P_Gaussian}
\end{align} 
Additionally, from Proposition~\ref{prop:worst_case} we see that the $Z = (X,Y)$ is an OU process under the worst-case measure $\P^*$,
\[d\begin{pmatrix}
    X_t \\ Y_t 
\end{pmatrix} = \begin{pmatrix}
    -\frac{1}{2}c_XM_XX_t  -\frac{1}{2}c_XM_YY_t \\
    \frac{1}{2}c_Y\beta_X X_t - \frac{1}{2}c_Y(\Sigma^{-1}_Y + (c_Y)^{-1}c_{YX}\Sigma^{-1}_{XY} + C\Sigma_{XY}(\Sigma_Y)^{-1})Y_t
\end{pmatrix}dt + c^{1/2}dW_t,\]
where we recalled the compatibility condition \eqref{eqn:beta_Y} for $\beta_Y$.

We now compare this setting and optimal strategy to the setup in \cite{itkin2025ergodic}.  The matrix $a(x)$ of \eqref{eqn:a} here is given by $c_Xp_X(x)$, where $p_X$ is the marginal density of $X$, so it is easy to see by inspection that $\widehat \phi(x) = \frac{1}{2}\log p_X(x)$ solves \eqref{eqn:Euler_Lagrange_x}. As such, the  robust growth-optimal strategy for the class $\Pi$ is given by 
\[\widehat \theta_t = \nabla \widehat \phi(X_t) = \frac{1}{2}\nabla \log p_X(X_t) =-\frac{1}{2}(\Sigma_X)^{-1}X_t.\]
To find the dynamics of the worst-case measure $\widehat \P$ we solve the PDE \eqref{eqn:v_PDE} for $\nabla_y \widehat v$, which has explicit solution given by
\begin{align} 
    \nabla_y \widehat v(x,y) & = \ell_Y^0(x,y) -  \frac{1}{2}(c_Y)^{-1}(\Sigma^{-1}_Y)^{-1}B^\top c_X(\Sigma^{-1}_X x +\Sigma^{-1}_{XY}y) \\
    & = - \frac{1}{2}(C^0x+D^0y) -  \frac{1}{2}(c_Y)^{-1}(\Sigma^{-1}_Y)^{-1}B^\top c_X(\Sigma^{-1}_X x +\Sigma^{-1}_{XY}y), \label{eqn:hatv_Gaussian}
\end{align}
where $C^0 = C-\beta_X$ and $D^0 = D-\beta_Y$ are the parts of $C$ and $D$ which are not affected by $\beta$.
We note that the right hand side for $\nabla_y \widehat v$ above is linear in $y$ and, as such, is indeed the gradient of some quadratic function $\widehat v$, which we do not write out explicitly. One can easily verify that \eqref{eqn:hatv_Gaussian} satisfies \eqref{eqn:v_PDE} and the verification follows in a very similar fashion to the computations in \eqref{eqn:u_verification_Gaussian}, where we verified that the explicit solution for ${\bf u}$ solves \eqref{eqn:u_Gaussian}, so we omit the verification here. 


It now follows from \eqref{eqn:hatP_dynamics} that under the worst-case measure $\widehat \P$, $Z$ follows a different OU process characterized by the dynamics
\[d\begin{pmatrix}
    X_t \\ Y_t
\end{pmatrix} = \begin{pmatrix}
    -\frac{1}{2}c_X(\Sigma_X)^{-1}X_t \\
    -\frac{1}{2}c_Y(C^0X_t +D^0Y_t) -\frac{1}{2}(\Sigma^{-1}_Y)^{-1}B^\top c_X(\Sigma^{-1}_X X_t +\Sigma^{-1}_{XY}Y_t)
\end{pmatrix}dt + c^{1/2}dW_t. \]
The growth rate of the robust-optimal strategy $\widehat \theta$ is the same under both $\P^*$ and $\widehat \P$, as both measures belong to $\Pi_0$. This growth rate is given by
\[\lambda_{\Pi} = g(\widehat \theta;\P^*) = g(\widehat \theta;\widehat \P) = \frac{1}{2}\int_{\R^d}\nabla \widehat \phi(x)^\top c_X \nabla \widehat \phi(x)p_X(x)dx = \frac{1}{8}\mathrm{Tr}((\Sigma_X)^{-1}c_X).\]
Conversely, despite $\theta^*$ having the growth-rate invariance property over the class $\Pcal$, this property does not extend to the larger class $\Pi$  as discussed in Section~\ref{sec:discussion}.  We are able to obtain by direct calculation that
\begin{align}
    g(\theta^*,\widehat \P) & = \int_{\R^{d+m}}\nabla_x\phi^*(z)^\top c_X\nabla\widehat \phi(x)p(z)dz - \frac{1}{2}\int_{\R^{d+m}}\nabla_x \phi^*(z)^\top c_X\nabla_x\phi^*(z)p(z)dz \\
    & = \frac{1}{4}\mathrm{Tr}(M_X^\top c_X) + \frac{1}{4}\mathrm{Tr}(M_Y^\top c_X(\Sigma_X)^{-1}\Sigma_{XY}) -  \lambda_{\Pcal},
\end{align}
where $\lambda_{\Pcal}$ was computed in \eqref{eqn:lambda_P_Gaussian}. We conclude this section by noting that $M_Y = 0$ if and only if 
\[\beta_X = -\Sigma^{-1}_{YX}- (c_Y)^{-1}c_{YX}\Sigma^{-1}_X - (c_Y)^{-1}(\Sigma^{-1}_Y)^{-1}B^\top c_X\Sigma^{-1}_X.\]
In this case it can be directly checked that $\theta^* = \widehat \theta$ and $\P^* = \widehat \P$ corresponding to the special
edge case mentioned in Section~\ref{sec:discussion} where knowledge of the stochastic factor does not increase the investor's robust growth rate. For all other choices of $\beta_X \in \R^{m \times d}$ the strategies $\theta^*$ and $\widehat \theta$ differ and the strict inequality $\lambda_\Pcal > \lambda_\Pi$ holds. 
\begin{remark}
Notice -- as a side result -- that our robust optimal growth rate problems under ergodicity maintain a linear (quadratic) character for drifts and strategies with respect to $\mathcal{P}$ and $\Pi$, respectively: their solutions equal the solutions of the corresponding linear quadratic problems formulated in a purely Gaussian universe.
\end{remark}

\subsection{Pairs trading}
\label{sec:pairs_trading}

\subsubsection{Pairs trading framework} \label{sec:pairs_trading_framework}
We briefly describe pairs trading strategies and how they can be embedded into our framework. Typically, such strategies trade on the \emph{spread} of two co-integrated securities. That is, if $S^1$ and $S^2$ are price processes of two cointegrated risky assets, a pairs trading strategy monitors the spread 
\begin{equation} \label{eqn:spread}
    X_t = aS^1_t - bS^2_t,
\end{equation} where $a$ and $b$ are constants chosen to ensure stationarity of the spread process. The ratio $a/b$ is often called the \emph{hedge ratio}, which we assume here to be constant. 

In its simplest instantiation, a pairs trading strategy bets on the spread process mean reverting to zero. That is, if the spread $X$ is positive the investor will short $S^1$ and long $S^2$, while if the spread is negative they will long $S^1$ and short $S^2$. More sophisticated pairs trading strategies take other factors into account when deciding the trading rule, such as allowing the mean reversion level of $X$ to vary stochastically. To compute the wealth dynamics of a pairs trading strategy we start by letting $Q_t$ be the number of units of $S^1$ that the investor holds at time $t$. The pairs trading approach then prescribes holding $-\frac{b}{a}Q_t$ units in $S^2$, with the trading activity financed by the risk-free asset. The wealth process then evolves according to
\[dV_t = Q_tdS^1_t - \frac{b}{a}Q_tdS^2_t = \frac{1}{a}Q_tdX_t.\]
As such, if we set $\theta_t = \frac{Q_t}{aV_t}$ then $V = V^\theta$, where $V^\theta$ is defined in \eqref{eqn:wealth}. The upshot is that one can solve the robust growth-optimization problem for the holdings $\theta$ by treating the spread $X$ as the one-dimensional asset process and then translate back to the holdings $Q$, invested in the original assets $S^1$ and $S^2$, via the transformation $Q_t = a\theta_tV_t^\theta$.
\subsubsection{Central Tendency Ornstein--Uhlenbeck model} \label{sec:pairs_CTOU}
In the pairs trading literature, models have been proposed to capture mean reversion properties of the spread process.\footnote{Some authors model the spread measured in logarithmic terms $a\log S^1_t - b\log S^2_t$, but since we wish to directly relate the spread to a tradeable security we work with the asset spread process \eqref{eqn:spread} instead.} A popular choice (see \cite{liu2017intraday,leung2018optimal}) is the so-called \emph{Central Tendency Ornstein--Uhlenbeck} (CTOU) model,
\begin{equation} \label{eqn:CTOU}
\begin{aligned}
    dX_t & = - \kappa_X(X_t-Y_t)dt&  + \sqrt{c_X} dW^X_t, \\
    dY_t & = -\kappa_YY_tdt&  + \sqrt{c_Y}dW^Y_t.
\end{aligned}
\end{equation}
Here $X$ mean reverts to the stochastic level $Y$ with mean reversion speed $\kappa_X > 0$ and $Y$ follows autonomous OU dynamics with mean reversion speed $\kappa_Y>0$. The volatility levels $\sqrt{c_X}$ an $\sqrt{c_Y}$ are assumed to be constant and the Brownian motions $W^X$ and $W^Y$ are uncorrelated. The bivariate process $Z = (X,Y)$ has $\Ncal(0,\Sigma)$ as its unique stationary measure, where
\begin{equation} \label{eqn:Sigma_pairs}
\Sigma = \begin{bmatrix} \frac{c_X}{2\kappa_X} + \frac{c_Y\kappa_X}{2\kappa_Y(\kappa_X+\kappa_Y)}& \frac{c_Y\kappa_X}{2\kappa_Y(\kappa_X+\kappa_Y)} \\ \frac{c_Y\kappa_X}{2\kappa_Y(\kappa_X+\kappa_Y)}
& \frac{c_Y}{2\kappa_Y}
\end{bmatrix}.
\end{equation}


Although a very useful model to inform portfolio selection, the CTOU model imposes strong modelling assumptions. In particular, the linear drift dynamics for $X$ directly imply that the growth-optimal portfolio is a linear feedback form function,
\[\theta^{\mathrm{CTOU}}_t = -\frac{\kappa_X}{c_X}(X_t-Y_t).\]

  We will apply the robust growth-optimal framework developed in this paper and in \cite{itkin2025ergodic} to study how robust the conclusions of the CTOU model are and what effects different assumptions about model uncertainty have on the strategy and the associated growth rate. To this end, we take a constant diagonal volatility matrix $c$ and a centered bivariate Gaussian invariant density $p$ with stationary covariance matrix $\Sigma$ given by \eqref{eqn:Sigma_pairs} as inputs, which puts us in the setting of Section~\ref{sec:OU} with $m=d=1$. These choices pin down the class of measures $\Pi$. To specify $\Pcal$ we additionally take $b_Y(y) = - \frac{\kappa_Y}{c_Y}y$ consistent with the CTOU model, which corresponds to $\beta_X = 0$ in the notation of Section~\ref{sec:OU}. 
  
  Since $Y$ is an autonomous one-dimensional process its  linear drift specification may be easier to statistically justify from data than the corresponding one for $X$. Moreover, if the investor uses exponentially weighted updates to their forecast $Y$, then the investor is effectively guaranteeing  (in the continuous-time limit) that the drift dynamics of $Y$ are as assumed. In this context, the class $\Pcal$ can be viewed as the mathematical idealization of expressing the investor's high confidence in the evolution of the autonomous process $Y$ relative to her lower confidence estimating the drift of the coupled process $X$. The class $\Pi$, in contrast, expresses equal uncertainty in the drift specifications of both $X$ and $Y$ and seeks to robustify over both those inputs.

Remarkably, the formulas of Example~\ref{sec:OU} show that $\theta^* = \theta^{\mathrm{CTOU}}$ and that the worst-case measure $\P^*$, for the class $\Pcal$, is a CTOU process. As such, the CTOU model can be viewed as a conservative modelling choice when considering measures that are consistent with a constant volatility matrix, centered Gaussian invariant density with covariance matrix $\Sigma$ of \eqref{eqn:Sigma_pairs} and autonomous OU dynamics for $Y$. In contrast, the strategy $\widehat \theta$ in this case is given by
\[\widehat \theta_t = - \frac{\kappa_X(\kappa_X +\kappa_Y)}{c_X(\kappa_X+\kappa_Y) + c_Y\kappa_Y}X_t\]
and the dynamics of $(X,Y)$ under $\widehat \P$, which is the worst-case measure for $\Pi$, are given by
\begin{align}
dX_t & = - \frac{c_X\kappa_X(\kappa_X +\kappa_Y)}{c_X(\kappa_X+\kappa_Y) + c_Y\kappa_Y}X_tdt & + \sqrt{c_X}dW^X_t, \\
dY_t & = \left(\frac{c_Y\kappa_X^2(\kappa_X+\kappa_Y)}{c_Y\kappa_X^2 + c_X\kappa_Y(\kappa_X+\kappa_Y)}X_t - \frac{(\kappa_X+\kappa_Y)(c_Y\kappa_X^2 + c_X\kappa_Y^2)}{c_Y\kappa_X^2 + c_X\kappa_Y(\kappa_X + \kappa_Y)}Y_t\right)dt & + \sqrt{c_Y}dW^Y_t.
\end{align}
The dynamics under $\widehat \P$ are consistent with the volatility matrix $c$ and invariant density $p$, but admits a different drift for $Y$ than the CTOU process does. The main feature of $\widehat \P$ is that it prescribes adversarial dynamics for $X$, making $X$ an autonomous diffusion. This is in stark contrast to the CTOU model \eqref{eqn:CTOU} where $Y$ was an autonomous diffusion and $X$ mean reverted to the stochastic target $Y$. It is evident that the strategy $\theta^*$ is suboptimal under $\widehat \P$ as it uses the superfluous factor $Y$, which does not appear in the dynamics of $X$, to make investment decisions. 

\begin{figure}
\centering
\includegraphics[width=0.75\linewidth]{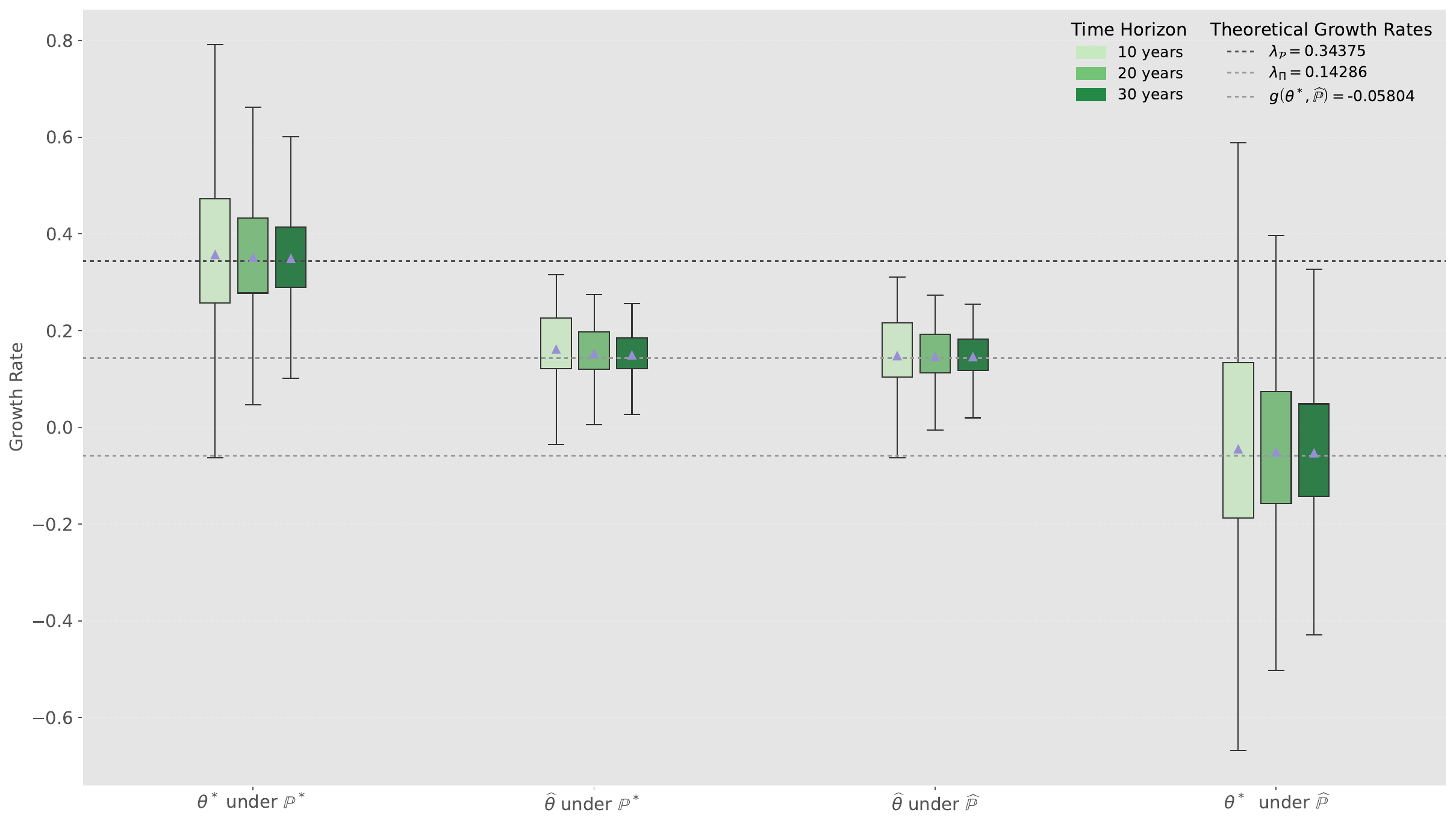}
    \caption{ Boxplots of growth rates for $\theta^*$ and $\widehat \theta$ under $\P^*$ and $\widehat \P$ obtained from 10,000 simulations with time horizon $T \in \{10,20,30\}$, increasing from light to dark green, with outliers omitted. The triangle in each box represents the mean and the dashed lines are the theoretical growth rates $\lambda_\Pcal, \lambda_{\Pi}$ and $g(\theta^*;\widehat \P)$ appearing in descending order.} \label{fig:OU_boxplot} \end{figure}

We now numerically illustrate the performance of the strategies. Our experiments use the following model parameters, in annual units, which are representative of a profitable pairs trading opportunity, 
\begin{equation} \label{eqn:OU_params} 
c_X = 0.04, \quad  c_Y = 0.0225, \quad \kappa_X = 1, \quad \kappa_Y = 1/2.
\end{equation}
Figure~\ref{fig:OU_boxplot} shows boxplots of the growth-rates $\frac{1}{T}\log V^\theta_T$  simulated up to the horizons $T \in \{10,20,30\}$ for $\theta \in \{\theta^*,\widehat \theta\}$ and under both worst-case measures $\P^*$ and $\widehat \P$.
As predicted by \eqref{eqn:growth_gap} and \eqref{eqn:growth_rate_differential}, the results show that the strategy $\theta^*$ outperforms $\widehat \theta$ under $\P^*$ and underperforms it under $\widehat \P$ by the same growth rate differential. The outperformance is substantial with $\lambda_{\Pcal} - \lambda_\Pi \approx 0.2$ showcasing the potential long-term benefits of utilizing the stochastic factor.  On the flip side, $g(\theta^*;\widehat \P) =-0.06$ leading not only to underperformance, but to investment losses when the underlying dynamics correspond to the worst-case measure $\widehat \P$ for the class $\Pi$. In contrast, the strategy $\widehat \theta$ is ambivalent to the underlying measure leading to a solid annual growth rate of approximately 0.14. Although not theoretically studied, the simulations also show that the growth-rate when using $\widehat \theta$ converges faster than its counterpart $\theta^*$ under both measures. 


Next, we compare the holdings prescribed by the two strategies. Figure~\ref{fig:OU_slices} plots the feedback form function $\widehat \phi'$ specifying the strategy $\widehat \theta$ together with the slices $\partial_x  \phi^*(\cdot,y)$, for certain fixed values of $y$, which specify the holdings $\theta^*$ when $Y$ takes the value y. We see that both strategies take short positions in $X$ when it is positive and long positions when it is negative, as expected for a pairs trading strategy. However, $\widehat \phi'$ is symmetric around zero, while $\partial_x\phi^*(\cdot,y)$ is symmetric around $y$ since the latter anticipates that the spread process $X$ is mean reverting to the level $y$. Additionally, although all of the slices $\partial_x \phi^*(\cdot,y)$ are parallel to each other, they are not parallel to $\widehat \phi'$. Indeed, the former has a steeper slope, which means that an investor using $\theta^*$ trades more aggressively than one using $\widehat \theta$ since an equal sized movement in $X$ leads to a larger change in the holdings $\theta^*$ than for its counterpart $\widehat \theta$. The less aggressive behaviour of $\widehat \theta$ is consistent with its role as the robust growth-optimal strategy under the larger class of measures $\Pi$. 
\begin{figure}
    \centering
    \includegraphics[width=0.75\linewidth]{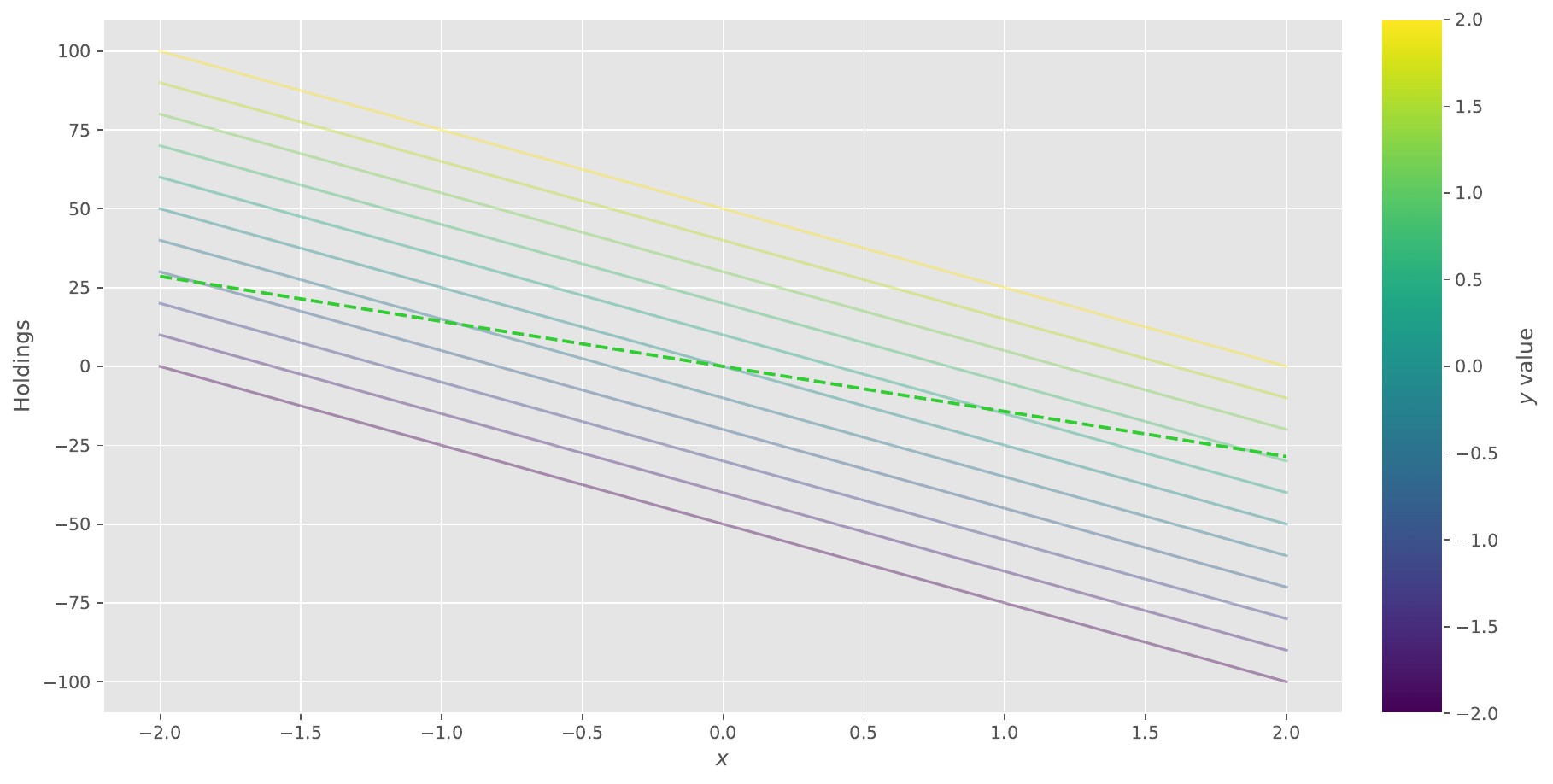}
    \caption{Slices of $\partial_x \phi^*(\cdot,y)$ for $y$ between $-2$ and $2$ for 11 equally spaced points (solid lines with y increasing from dark to light) plotted along side $\widehat \phi'$ (dashed line).}
    \label{fig:OU_slices}
\end{figure}

\subsubsection{Fat-tails} \label{sec:fat_tails}
In this section we explore an extension to invariant distributions with fat tails. Indeed, it is well-documented that asset returns distributions are fat-tailed (see e.g.\ \cite{cont2001empirical}) and, as such, it may be of interest to relax the Gaussianity assumption of the previous section. To this end we continue to work with a constant diagonal volatility matrix $c = \mathrm{diag}(c_X,c_Y)$, but now assume that the invariant distribution of $(X,Y)$ is a bivariate $t$-distribution,
\[p(z) = \frac{\Gamma\big(\frac{\nu + 2}{2}\big)}{\Gamma(\frac{\nu}{2})\nu\pi\sqrt{\det \Sigma}}\left(1 + \frac{1}{\nu}z^\top \Sigma^{-1}z\right)^{-\frac{\nu + 2}{2}}.\]
Here, as before, $\Sigma$ is a symmetric positive-definite matrix and $\nu > 1$ is the degrees-of-freedom parameter. For the sake of consistency with the previous subsection, we seek a choice for the input $b_Y$ that ensures $Y$ has autonomous dynamics; that is, we require it to be a function of $y$ only. It is straightforward to establish that the unique such choice satisfying the compatibility condition \eqref{eqn:compatibility_condition} and ensuring that $\ell_Y$ satisfies Assumption~\ref{ass:conditions}\ref{item:finite_growth} is
\[b_Y(y)  = \frac{1}{2}\frac{\int_{-\infty}^\infty\partial_y p(x,y)dx}{\int_{-\infty}^\infty p(x,y)dx} = \frac{1}{2}\frac{\partial}{\partial y}\log\left(\int_{-\infty}^\infty p(x,y)dx\right) = \frac{1}{2}(\log p_Y)'(y) = -\frac{\nu + 1}{2(\nu\Sigma_Y + y^2)} y,\]
where $p_Y$ is the marginal distribution of $Y$, which is known to be a univariate location-scale $t$-distribution with $\nu$ degrees of freedom, location parameter zero and scale parameter $\sqrt{\Sigma_{Y}}$. 

With the inputs fixed we now directly solve \eqref{eqn:div_PDE}  by integration to obtain
\begin{align}
{\bf u}(x,y) & = \partial_y\left(\int_{-\infty}^x c_Y\ell_Y(x',y)p(x',y)dx'\right) = \frac{c_Y}{2}\partial_y\left(\int_{-\infty}^x \partial_y p(x',y)dx' - 2b_Y(y)\int_{-\infty}^xp(x',y)dx'\right) \\
& = \frac{c_Y}{2}\partial_y\left(\int_{-\infty}^x \Big(\partial_y(p_Y(y)p_{X|Y}(x'|y)) - p_Y'(y)p_{X|Y}(x'|y)\Big)dx'\right) \\ & = \frac{c_Y}{2}\partial_y\left(p_Y(y)\int_{-\infty}^x \partial_y p_{X|Y}(x',y)dx'\right) = \frac{c_Y}{2}\partial_y(p_Y(y)\partial_y F_{X|Y}(x|y)),
\end{align}
where $F_{X|Y}(\cdot|y)$ is the CDF of $X|Y=y$. It is similarly known that $X|Y=y$ is an a location-scale $t$-distribution with $\nu+1$ degrees of freedom and with location parameter $\mu(y)$ and scale parameter $\tau(y)$ given by 
\[\mu(y) = \frac{\Sigma_{XY}}{\Sigma_Y}y, \qquad \tau^2(y) = \frac{( y^2+\Sigma_Y\nu)\det \Sigma}{(\nu + 1)\Sigma_Y^2}. \]
By centering and normalizing we obtain that $F_{X|Y}(x|y) = F_{\nu + 1}(\frac{x-\mu(y)}{\tau(y)})$, where $F_{\nu+1}(\cdot)$ is the CDF of a univariate $t$-distribution with $\nu+1$ degrees of freedom. As such, ${\bf u}$ can be explicitly computed and the expression involves the density of the $t$-distribution and its derivatives, but due to its length we omit it here.
The optimal strategy $\theta^*$ is then given from \eqref{eqn:gradient_case} as
\[\theta^*_t = \frac{1}{2}\partial_x \log p(X_t,Y_t) + \frac{c_Y}{2c_X}\frac{\partial_y \left(p_Y(Y_t)\partial_yF_{\nu+1}(\frac{X_t-\mu(Y_t)}{\tau(Y_t)})\right)}{p(X_t,Y_t)},\]
which can similarly be computed explicitly in closed form.\footnote{The analytic formula can be found in our code on GitHub.}

For the $Y$-unconstrained problem we first note that $a(x) = c_Xp_X(x)$, so from \eqref{eqn:Euler_Lagrange_x} we see that 
\[\widehat \theta_t = - \frac{1}{2}(\log p_X)'(X_t) = -\frac{(\nu + 1)X_t}{2(\nu \Sigma_X + X_t^2)}. \]
To better understand the differences between the two strategies we plot the slices $\partial_x \phi^*(\cdot,y)$ and  $\widehat \phi'$ in Figure~\ref{fig:tdist_slices}, akin to Figure~\ref{fig:OU_slices} which showcased the analogue for the Gaussian example. For ease of comparison we use the same $c$ and $\Sigma$ matrices here as in the Gaussian example, which are specified by \eqref{eqn:OU_params} and \eqref{eqn:Sigma_pairs}. The main difference then lies in the degrees-of-freedom parameter, which we take to be $\nu = 3$ so as to invoke a fat-tailed invariant measure.

\begin{figure}
    \centering
    \includegraphics[width=0.75\linewidth]{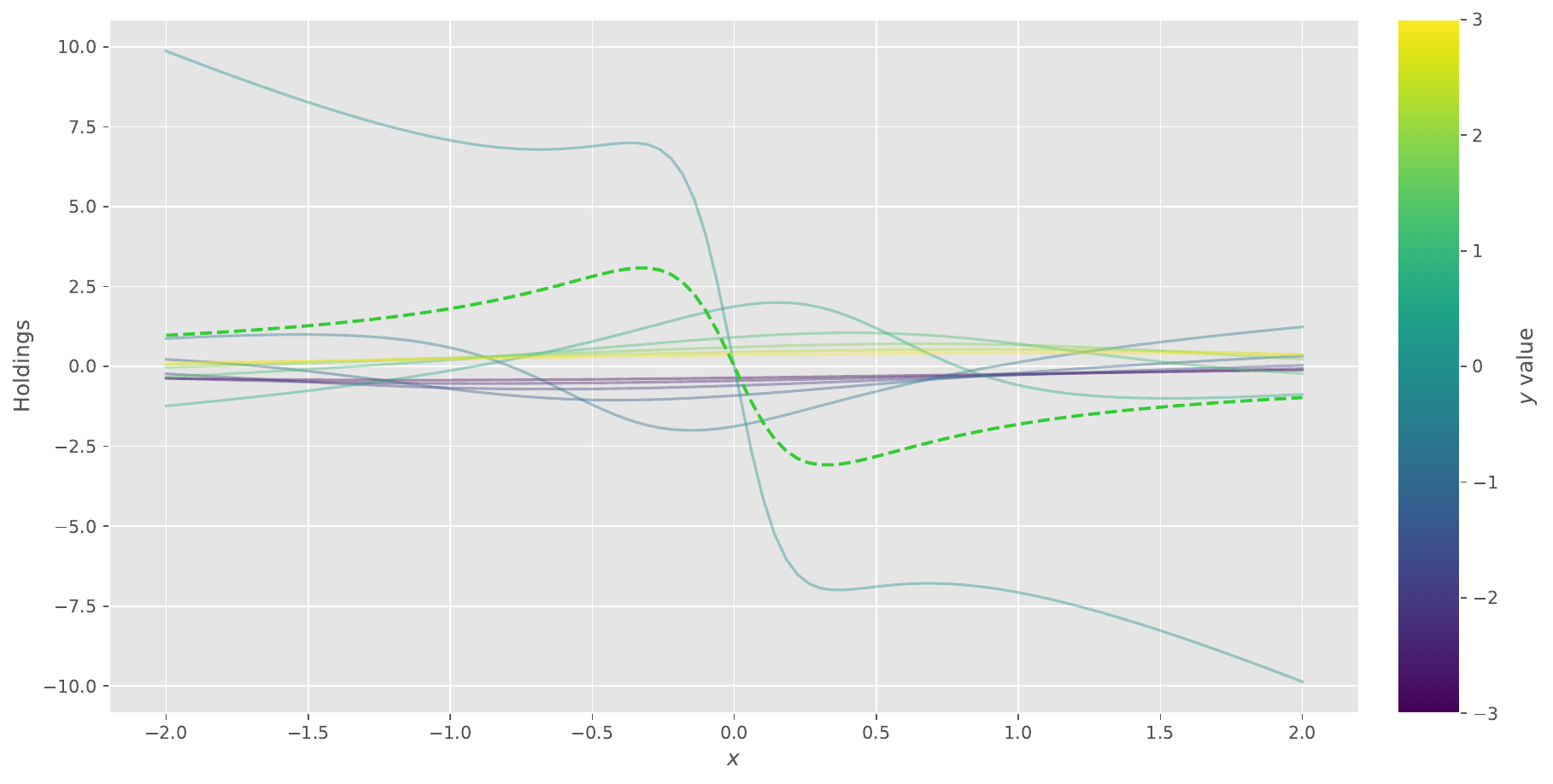}
  \caption{Slices of $\partial_x \phi^*(\cdot,y)$ for $y$ between $-3$ and $3$ for 11 equally spaced points (solid lines with y increasing from dark to light) plotted along side $\widehat \phi'$ (dashed line).}    \label{fig:tdist_slices}
\end{figure}

As in the Gaussian case, both $\widehat \phi'$ and $\partial_x \phi^*(\cdot,0)$ are symmetric around zero and the strategy $\theta^*$ takes larger positions (in absolute value) when $y = 0$ than $\widehat \theta$ does for the same spread process value $x$, reflecting its more aggressive trading tendencies. However, the net position taken for both strategies away from the mean level $x = 0$ is an order of magnitude smaller than prescribed in the Gaussian example. Due to the fat tails of the invariant distribution, the investor limits their position sizes to avoid the more likely risk of large losses. As in the Gaussian case, when $y$ differs from its stationary mean level of zero, the strategy $\theta^*$ is no longer symmetric around zero, but rather skews towards the value of $y$. Differently, however, the magnitude of the position decreases substantially when $y$ takes on a larger value reflecting the investor's caution due to the fact that large values of $Y$ may persist for longer periods of time because of the presence of fat tails.     

Another new feature present in this example is the nonlinear and nonmonotonic behaviour of the strategies. Although both strategies takes long positions when $X$ falls below its stochastic mean reversion level and short positions otherwise, the size of the position does not always increase with the spread level. This is most clearly illustrated by both $\widehat \phi'$ and the $\phi^*(\cdot,0)$ slice. As $x$ initially becomes negative (the positive side is symmetric) the investor increases their position in a steep, approximately linear, fashion. This continues until a critical point is reached, after which the investor maintains their long position, but reduces it relative to the peak level. In the case of $\widehat \theta$, the reduction continuous indefinitely, while the strategy $\theta^*$ starts increasing the holdings again after a further threshold is reached. 

This more complicated behaviour can be attributed to the effect the fat tails of the invariant measure have on the dynamics of $(X,Y)$. When $|X|$ is small, the investor is happy to acquire an increasingly large pairs trading position to benefit from mean reversion tendencies. However, as $X$ starts to grow in magnitude, the process with relatively high probability may locally maintain its value, or even continue to increase further, so the investor hedges their bet by reducing their position. In the case of full robustness over $X$ and $Y$ the investor continuous to reduce their position indefinitely as $|X|$ grows. Conversely, in the case when the dynamics of $Y$ are known, at a further threshold, the investor infers from the additional information available to them that the fat-tailed event they are witnessing is beyond a typical occurrence for the process $(X,Y)$ and starts to increase their holdings once more.

\subsubsection{Stochastic volatility} \label{sec:stoch_vol}
In this section we explore how our results can incorporate stochastic volatility. To this end, rather than using $Y$ to model the stochastic mean reversion level, we will use it to model the stochastic volatility of $X$. Concretely, we assume that the volatility of $X_t$ is given by $\sqrt{Y_t}$ and we model $Y$ as an autonomous Cox--Ingersoll--Ross (CIR) process,
\[dY_t = \kappa(\nu - Y_t)dt + \sigma\sqrt{Y_t} dW^Y_t.\]
Here $\nu > 0$ is the long-run mean reversion level for the volatility process, $\kappa > 0$ dictates the speed of mean reversion, $\sigma > 0$ is the vol-of-vol parameter and the Brownian motion $W^X$ driving $X$ is correlated with $W^Y$  with correlation coefficient $\rho \in (-1,1)$. In financial modelling, the CIR process appears in the Heston model \cite{heston1993closed} and, although mostly used for options pricing purposes, there are well-established calibration procedures to estimate the model parameters. When the Feller condition $2\kappa \nu > \sigma^2$ holds, it is well-known that $Y$ takes values in $D = (0,\infty)$ and its stationary distribution is $\Gamma(\alpha,\beta)$, where
\[\alpha = \frac{2\kappa\nu}{\sigma^2}, \qquad \beta = \frac{2\kappa}{\sigma^2}.\]
To finalize the inputs, it just remains to encode the long-run behaviour of $X$. Since $\sqrt{Y}$ represents its volatility we will assume that conditional on $Y=y$, the long-run distribution of $X$ is $\Ncal(0,y)$. That fixes our inputs to be
\[c(y) = \begin{bmatrix}
    y & \rho \sigma y \\ \rho \sigma y & \sigma^2 y
\end{bmatrix}, \qquad p^{\alpha,\beta}(x,y) =             \frac{1}{\sqrt{2 \pi y}} \exp \left( -\frac{x^2}{2y} \right)
     \frac{\beta^\alpha}{\Gamma(\alpha)}y^{\alpha-1}e^{-\beta y}, \qquad b_Y(y) = \frac{\kappa(\nu - y)}{\sigma^2 y},\]
     where we emphasize the parameters $\alpha,\beta$ in the invariant density notation.
This leads to 
\begin{align}
    \ell_X(x,y)  =  \frac{\rho\sigma x^2}{4y^2} - \frac{x}{2y} + \frac{4\kappa \nu \rho - \rho \sigma^2}{4\sigma y} - \frac{\kappa \rho}{\sigma}, \qquad \ell_Y(x,y) = \frac{x^2}{4y^2} - \frac{1}{4y} - \frac{\rho x}{2\sigma y}
\end{align}
and it is straightforward to check that the compatibility condition \eqref{eqn:compatibility_condition} holds with these inputs. 

We now derive the optimal strategies. For the robust problem over $\Pcal$ we begin by solving \eqref{eqn:div_PDE}, which has explicit solution
\[{\bf u}(x,y) = \left(\frac{-\sigma ^3 x^3+2 \rho  \sigma ^2 x^2 y+4 \kappa  \sigma  xy^2+(3 \sigma ^3-4 \kappa  \nu  \sigma)xy-8 \kappa  \rho  y^3+ (8 \kappa  \nu  \rho -2 \rho  \sigma ^2)y^2}{8 \sigma  y^2}\right)p(x,y)\]
obtained by integration and direct calculation. From \eqref{eqn:gradient_case} we then have that 
\begin{equation} \label{eqn:theta*_stochvol}
    \theta^*_t = \frac{-\sigma ^3 X_t^3+4 \rho  \sigma ^2 X_t^2 Y_t+(4 \kappa  \sigma -4 \sigma )X_tY_t^2+(3 \sigma ^3-4 \kappa  \nu 
   \sigma)X_tY_t -16 \kappa  \rho  Y_t^3+ (16 \kappa  \nu  \rho -4 \rho  \sigma ^2)Y_t^2}{8\sigma Y_t^3}.
\end{equation}
   For the $\Pi$ problem, we note that $a(x) = \int_0^\infty yp^{\alpha,\beta}(x,y)dy$ and $\partial_x p^{\alpha,\beta}(x,y) = -\int_0^\infty \frac{x}{y}p^{\alpha,\beta}(x,y)dy$. As such, from \eqref{eqn:Euler_Lagrange_x} we see that
   \[\widehat \theta_t = \frac{1}{2}(\log a)'(X_t) = -\frac{X_t}{2}\frac{\int_0^\infty p^{\alpha,\beta}(X_t,y)dy}{\int_0^\infty yp^{\alpha,\beta}(X_t,y)dy} = - \frac{\alpha}{2\beta} X_t\frac{p_X^{\alpha,\beta}(X_t)}{p_X^{\alpha+1,\beta}(X_t)},\]
   where $p_X^{\alpha,\beta}(x) = \int_0^\infty p^{\alpha,\beta}(x,y)dy$ is the marginal density of $X$. 

   The left panel of Figure~\ref{fig:stochvol_slices} shows the slices $\partial_x \phi^*(\cdot,y)$ for a range of values of $y$ together with $\widehat \phi'$, while the right panel plots $\widehat \phi'$ on its own. The model parameters here are chosen to be
   \[\kappa = 5, \qquad \nu = 0.04, \qquad \sigma = 0.6, \qquad \rho = 0.\] 
From the left panel we see that $\theta^*$ typically takes much larger positions than $\widehat \theta$. The strategy $\theta^*$ benefits from knowing the stochastic volatility level and the local dynamics of $Y$, allowing the investor to confidently take larger positions. As $y$ decreases, the absolute position sizes $|\theta^*|$ increase as the lower volatility level leads to a larger signal-to-noise ratio. By inspecting \eqref{eqn:theta*_stochvol} we see that, to leading order, the increase happens at a rate that is inversely proportional to $y^3$ indicating that $\theta^*$ achieves its growth-rate outperformance by trading very aggressively in low volatility environments.

\begin{figure}
    \centering
\includegraphics[width=1\linewidth]{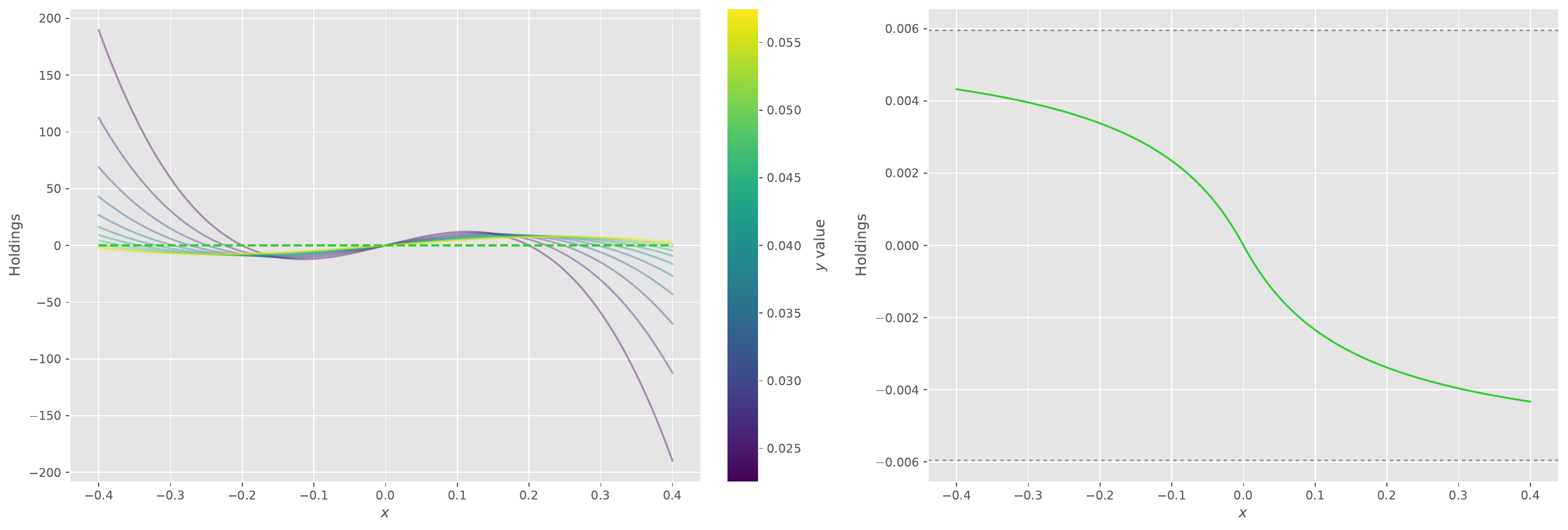}
\caption{The left panel plots slices of $\partial_x \phi^*(\cdot,y)$ for $y$ between $0.0225$ and $0.0575$ for 11 equally spaced points  (solid lines with y increasing from dark to light) plotted along side $\widehat \phi'$ (dashed line). The right panel depicts only $\widehat \phi'$ (solid line) together with the theoretical limiting holdings of the strategy $\lim_{x \to \pm \infty} \widehat \phi'(x) =  \mp \frac{\alpha^2}{\sqrt{2}\beta^{3/2}}$ (dashed lines).}    \label{fig:stochvol_slices}
\end{figure}

Surprisingly, the function $\partial_x \phi^*(\cdot,y)$ is cubic with a positive linear term for typical parameter values and, as such, contains a region close to zero, which prescribes the investor to take a position of the same sign as $X$ -- that is to bet that the magnitude of the spread will locally \emph{increase}. Once the magnitude of the spread reaches a large enough threshold the cubic term in the numerator of \eqref{eqn:theta*_stochvol} starts to dominate and the strategy takes on a more standard pairs trading form by taking a short position when the spread is positive and vice versa when the spread is negative. This surprising change in sign for the holdings can be attributed to the robustification over the drift of $X$. Although the investor views $Y$ purely as the stochastic volatility, in the worst-case measure $\P^*$, $Y$ also enters in the drift of $X$ playing a dual role. By inspecting the drift of $X$, it is evident that, for our parameter values and typical values of $y$, the drift of $X$ and $X$ itself share the same sign near zero. This suggests that, locally, the magnitude of the spread is likely to continue increasing, despite the fact that its long-run mean is zero. Moreover, the growth-rate invariance property of $\theta^*$ over $\Pcal_0$ guarantees that this strategy continues to perform well in every measure compatible with the three inputs $(c,p,b_Y)$. This surprising behaviour of the strategy highlights the importance of incorporating model uncertainty and illustrates how robust optimization can reveal roles the stochastic factor may play beyond its intended one.

From the scale of the left panel of Figure~\ref{fig:stochvol_slices}, the strategy $\widehat \theta$ appears flat due to it prescribing holdings orders of magnitude smaller than $\theta^*$ does. Here the class $\Pi$ essentially only encodes information about the invariant density $p$, with limited information about the volatility of $X$ known, since $c_X(Z_t) = Y_t$ and the dynamics of $Y$ are not fixed. This uncertainty over the dynamics of the stochastic volatility process leads to the observed conservative behaviour of $\widehat \theta$. The right panel of Figure~\ref{fig:stochvol_slices} provides a zoomed in look at the behaviour of this strategy. We see that near $x = 0$, the strategy prescribes a typical pairs trading position acquired at a near linear rate, qualitatively similar to the Gaussian case of Section~\ref{sec:pairs_CTOU}. However, as the spread continues to grow the holdings plateau, further reflecting the conservative approach it prescribes. Indeed, unlike its counterpart $\theta^*$, the holdings that $\widehat \theta$ prescribes are bounded with explicitly computable limit $\lim_{x \to \pm \infty} \widehat \phi'(x) = \mp \frac{\alpha^2}{\sqrt{2}\beta^{3/2}} \approx \mp 0.006$, which appear as dashed lines in the right panel of Figure~\ref{fig:stochvol_slices}. 

\section{Conclusion}
\label{sec:conclusion}
This paper studied a natural robust asymptotic growth optimization problem, where the quadratic variation and invariant density of $(X,Y)$ were fixed together with the drift of the stochastic factor process $Y$. We derived the robust growth-optimal strategy $\theta^*$, which is of gradient feedback-form type characterized by a function $\phi^*$ satisfying the Euler--Lagrange equation \eqref{eqn:Euler-Lagrange}, as well as the corresponding robust optimal growth rate $\lambda_{\Pcal}$ in a surprisingly explicit form, see Theorem \ref{thm:main}. This setup assumed more information about the dynamics of $Y$, through the input $b_Y$, than the previous study \cite{itkin2025ergodic}. This led to the optimal strategy $\theta^*$ depending, in a feedback form way, on $Y$, which the optimal strategy $\widehat \theta$ of \cite{itkin2025ergodic} did not depend on. In other words: non-traded factors matter if one has good knowledge about them.  
The ensuing examples showed that the full multidimensional problem is explicitly solvable in a Gaussian environment. We then used our robust framework to study a low-dimensional pairs trading application by robustifying the popular CTOU model as well as exploring extensions which accommodate fat tails and stochastic volatility. We showcased that knowledge of the stochastic factor can lead to strict improvement in growth rate of the strategy $\theta^*$ over $\widehat \theta$. However, if the investor is overconfident in the dynamics of the stochastic factor and the true drift does not coincide with what the input $b_Y$ specifies, then the strategy $\theta^*$ may substantially underperform $\widehat \theta$. An important problem for future work is to extend these results to other optimality criteria, such as power or exponential utility from wealth and/or consumption, so that the investor's risk-aversion can be incorporated into this framework. Exploring the finite horizon problem in detail is another interesting and important extension.

\appendix
\section{Proofs} \label{sec:proofs}
In this section we will make use of the following notation. We let $(E_n; n \in \N)$ be a sequence of increasing open sets with $C^1$ boundary, relatively compact in $E$ that exhaust $E$. That is, each $E_n$ is open, $\overline E_n \subset E_{n+1}$ and $E =  \bigcup_n E_n$.  We let $(D_n;n\in \N)$ serve an analogous role for $D$ and set $F_n = E_n \times D_n$ for every $n$. In the sequel we will make use of the $L^p$ and Sobolev spaces $L^p(U;V)$ and $W^{k,p}(U;V)$ respectively, for open domains $U$, vector spaces $V$, $p \geq 1$ and $k \in \N$. $L^p_{\mathrm{loc}}(U;V)$ denotes all $f: U \to V$ such that $f|_{K} \in L^p(K;V)$ for every $K \subset U$ compact and $W^{k,p}_{\mathrm{loc}}(U;V)$ is defined analogously. When $V  = \R$ we drop the range from the notation.

\subsection{Proof of Lemma~\ref{lem:variational}} \label{sec:proofs_variational}
For fixed $y \in D$ and any weakly differentiable $v:E \to \R$ we set
\[\|v\|_{\Wcal(y)} = \left(\int_E \nabla v(x)^\top c_X(x,y)\nabla v(x)p(x,y)dx\right)^{1/2}.\] We also define the equivalence relation $v \sim w \iff v-w$ is constant and denote the corresponding equivalence classes by $[v]$. We now define 
\[\Wcal(y) = \{ [v]: v \in W^{1,2}_{\mathrm{loc}}(E),\|v\|_{\Wcal(y)} < \infty\}.\]
It is easy to see that $\Wcal(y)$ is an inner product space when equipped with the inner product $(v,w)_{\Wcal(y)} = \int_E \nabla v(x)^\top c_X(x,y)\nabla w(x)p(x,y)dx$. In fact, $\Wcal(y)$ is a Hilbert space, which we now demonstrate. For any $n$ and any $v \in W^{1,2}(E_n)$ we have by the Poincar\'e inequality \cite[Theorem~5.8.1]{evanspartial2010} that 
\[\|v-v_{E_n}\|_{L^2(E_n)} \leq C_n\|\nabla v\|_{L^2(E_n)} \leq C_n'\|v\|_{\Wcal(x)},\]
where $v_{E_n} = \frac{1}{|E_n|}\int_{E_n}v(x)dx$, $C_n$ is the Poincar\'e inequality constant depending only on $E_n$ and $C_n' = C_n\epsilon_n^y$, where $\epsilon_n^y = 1/\inf_{x \in E_n} \{\lambda_{\min}(c_X(x,y))p(x,y)\}$.
Similarly we have that 
\[\|\nabla (v-v_{E_n})\|_{L^2(E_n)} = \|\nabla v\|_{L^2(E_n)} \leq \epsilon_n^y \|v\|_{\Wcal(y)}.\] As such it follows that if $(v^m)_{m \in \N}$ is a Cauchy sequence in $\Wcal(y)$ then $(v^m - v^m_{E_n})_{m \in \N}$ is a Cauchy sequence in $W^{1,2}(E_n)$ and hence has a limit $v_n$. Moreover, since $E_n \subset E_{n+1}$ we see that $v^m - v^m_{E_{n+1}} \to v_{n+1}$ in $L^2(E_n)$ as well. By writing $v^m - v^m_{E_{n+1}} = v^m - v^m_{E_{n}} + v^m_{E_{n}} - v^m_{E_{n+1}}$ we deduce that
\[v_{n+1} = v_n + C(n,n+1) \quad \text{on }E_n\]
for some constant $C(n,n+1)$. This now allows us to define for almost every $x \in E,$
\[v(x) = v_n(x) - \sum_{k=1}^{n-1} C(k,k+1); \quad \text{if }x \in E_n \text{ for } n=1,2,\dots\]
It is straightforward to check (in a similar way to the proof of Lemma~A.2 in \cite{kardaras2021ergodic}) that $v$ is well-defined. The above construction shows that $v^m \to v$ in $W^{1,2}_{\mathrm{loc}}(E)$. As such $\nabla v^m$ converges to $\nabla v$ almost everywhere. Hence, by Fatou's lemma we have that 
\[\|v^m - v\|_{\Wcal(y)} \leq \liminf_{k \to \infty}\|v^m-v^k\|_{\Wcal(y)},\]
which shows that $\lim_{m \to \infty} \|v^m- v\|_{\Wcal(y)} = 0$ since $(v^m)_{m \in \N}$ is Cauchy. This establishes completeness of $\Wcal(y)$. 
We now consider the subspace \[\Wcal_0(y) = \overline{\{\phi \in C_c^\infty(E)\}}^{\Wcal(y)}.\]
and establish a preliminary technical lemma, which seeks to solve a variational problem over the space $\Wcal_0(y)$.
\begin{lem} \label{lem:variational_technical}

    For every $y \in D$ there exists a unique solution $\phi^*(\cdot,y) \in \Wcal_0(y)$ to the variational problem   \begin{equation} \label{eqn:variational_y}
     \min_{\phi \in \Wcal_0(y)} J^y(\phi), \quad \text{where} \quad  J^y(\phi) = \int_E (\nabla \phi(x) - \xi(x,y))^\top c_X(x,y)(\nabla \phi(x) - \xi(x,y))p(x,y)dx.
    \end{equation}
    Moreover, one can select a version, which we again label $\phi^*$, such that $(x,y) \mapsto \phi^*(x,y)$ is measurable and $\phi^*(\cdot,y) \in C^2(E)$ for a.e.\ $y\in D$. Additionally, $\phi^*$ satisfies the PDE \eqref{eqn:Euler-Lagrange} and we have that $\nabla_x \phi^* \in L^q_{\mathrm{loc}}(F;\R^d)$ for every $n \in \N$ and $q \in [2,\infty)$. 
\end{lem}


\begin{proof} We first fix $y$ and note that $J^y(\phi)$ is well defined for any $\phi \in \Wcal_0(y)$ as the value of the integral is the same for any representative of the equivalence class $[\phi]$. In the course of this proof we well write $|\xi|^2_{\Wcal(y)}$ for $\int_E \xi(x,y)^\top c_X(x,y)\xi(x,y)p(x,y)dx$ to simplify the exposition. Now by the reverse triangle inequality we have that $\sqrt{J^y(\phi)} \geq \|\phi\|_{\Wcal(y)} - |\xi|_{\Wcal(y)}$, or equivalently that
\begin{equation} \label{eqn:phi_wbound}
    \|\phi\|_{\Wcal(y)} \leq \sqrt{J^y(\phi)} + |\xi|_{\Wcal(y)}
\end{equation}
for any $\phi \in \Wcal_0(y)$. Next we set $\hat J^y := \inf_{ \phi \in \Wcal_0(y)} J^y(\phi)$, which is clearly nonnegative. Letting $(\phi_n)_{n \in \N}$ be a sequence in $\Wcal_0(y)$ approaching $\hat J^y$, we see from \eqref{eqn:phi_wbound} that $(\phi_n)_{n \in \N}$ is a norm bounded sequence in $\Wcal_0(y)$. As such there exists a subsequence that converges weakly to some $\phi^*(\cdot,y) \in \Wcal_0(y)$. Since $J^y$ is weakly lower semicontinuous we see that
\[J^y(\phi^*(\cdot,y)) \leq \liminf_{n \to \infty} J^y(\phi_n) = \hat J^y\] so that $\phi^*(\cdot,y)$ is a minimizer. It is clear that it is the unique minimizer by the strict convexity of $ J^y$. 
Next, we will show that $\phi^*(\cdot,y)$ is a weak solution to \eqref{eqn:Euler-Lagrange_xi}. To see this let $\psi \in C_c^\infty(E)$ be arbitrary and note that for any $\epsilon > 0$ we have by optimality of $\phi^*$ that
\begin{align} 
 0 & \geq \frac{1}{\epsilon}\big(J^y(\phi^*(\cdot,y)) - J^y(\phi^*(\cdot,y) \pm \epsilon \psi)\big) \\
 & = \pm 2 \int_E\nabla \psi(x)^\top c_X(x,y)(\nabla_x \phi^*(x,y) - \xi(x,y))p(x,y)dx \pm \epsilon \int_E \nabla \psi(x)^\top c_X(x,y)\nabla \psi(x)p(x,y)dx. 
\end{align}
Sending $\epsilon \to 0$ shows that
\[0 = \int_E \nabla \psi(x)^\top c_X(x,y)(\nabla_x \phi^*(x,y) - \xi(x,y))p(x,y)dx ; \qquad \forall \psi \in C_c^\infty(E),\]
which is precisely the weak formulation of \eqref{eqn:Euler-Lagrange_xi}. From Assumption~\ref{ass:conditions}\ref{item:divergence} we see that the weak formulations of \eqref{eqn:Euler-Lagrange_xi} and \eqref{eqn:Euler-Lagrange} are the same so that $\phi^*$ also solves \eqref{eqn:Euler-Lagrange} weakly.

That $\phi^*(\cdot,y) \in C^2(E)$, that it is a strong solution to \eqref{eqn:Euler-Lagrange} and that a jointly measurable in $(x,y)$ version can be selected now follows in exactly the same way as Lemma~A.8 and Theorem~B.1 in \cite{itkin2025ergodic} respectively. Indeed, the proofs of those results used local arguments working in the interior of the domain, which is entirely unaffected from the slightly different definition of the space $\Wcal_0(y)$ used here. We exclusively use this jointly measurable version $\phi^*$ going forward. 

The proof will be complete once we argue that $\nabla_x \phi^* \in L^q(F_n;\R^d)$ for every $q \geq 2$ and $n \in \N$. This final part of the proof follows in a similar fashion to \cite[Theorem~A.3]{itkin2025ergodic}. For the remainder of the proof we fix $n$ and set 
\[\delta = \mathrm{dist}(E_n,\partial E) \land 1,\]
where by convention the distance is infinity if $\partial E = \emptyset$. Next we pick balls $(B_{\delta/4}(x_i))_{i=1}^N$ of radius $\delta/4$ centered at some $x_1,\dots,x_N \in E_n$ which cover $E_n$. We note that the number of balls $N$ depends only on $\delta$ and $|E_n|$, and by construction $B_{\delta/4}(x_i) \subset  B_{\delta/2}(x_i) \subset E$ for every $i$. Moreover, there exists some $n' > n$ such that
\[E_n \subset \bigcup_{i=1}^N B_{\delta/4}(x_i)  \subset \bigcup_{i=1}^N B_{\delta/2}(x_i) \subset E_{n'}.\]
Next we define $u_n(x,y) = \phi^*(x,y) - \phi^*_{n'}(y)$, where $\phi^*_{n'}(y) = \frac{1}{|E_{n'}|}\int_{E_{n'}}\phi^*(x,y)dx$. Since $u_n(\cdot,y)$ solves \eqref{eqn:Euler-Lagrange_xi} in $E$, it also solves it in $B_{\delta/4}(x_i)$ so we can apply \cite[Theorem~8.8]{gilbargelliptic2001} to obtain that $u_n(\cdot,y) \in W^{2,2}(B_{\delta/4}(x_i))$ and that
\[\mathrm{div}_x(c_X(x,y)\nabla_x u_n(x,y)p(x,y)) = \mathrm{div}_x(c_X(x,y)\xi(x,y)p(x,y)), \quad \text{for a.e.\ } (x,y) \in  B_{\delta/4}(x_i) \times D.\]
Next we can apply \cite[Theorem~11.2.3]{krylovsobolev2008} to obtain that $u_n \in W^{2,q}(B_{\delta/4}(x_i))$ and we have the estimate
 \[\|u_n(\cdot,y)\|_{W^{2,q}(B_{\delta/4}(x_i))} \leq C_y'\big(\|\mathrm{div}_x(c_X(\cdot,y)\xi(\cdot,y)p(\cdot,y))\|_{L^2(B_{\delta/2}(x_i))} + \|u_n(\cdot,y)\|_{L^2(B_{\delta/2}(x_i))}\big), \]
 where we set $a = c_X(\cdot,y)p(\cdot,y), b = c = d =0, f = \mathrm{div}_x(c_X(\cdot,y)\xi(\cdot,y)p(\cdot,y)), r = \delta/4$ and $p=2$ in the notation of \cite{krylovsobolev2008}. Here $C_y'$ is a constant that only depends continuously on $y$ through the positive values
 \[ \inf_{x \in E_{n'}}\lambda_{\min}(c_X(x,y))p(x,y), \qquad \|c_X(\cdot,y)p(\cdot,y)\|_{L^\infty(E_n')}, \qquad \|\mathrm{div}_x(c_X(\cdot,y)p(\cdot,y))\|_{L^\infty(E_n';\R^d)}.\]
 As such, by Assumption~\ref{ass:inputs} on the inputs $c_X$ and $p$, we have that $C' := \sup_{y \in D_{n'}} C'_y < \infty$.
 Now by summing over $i$ we obtain the estimate
 \begin{align} \|u_n(\cdot,y)\|_{W^{2,q}(E_n)} &  \leq \sum_{i=1}^N \|u_n(\cdot,y)\|_{W^{2,q}(B_{\delta/4}(x_i))} \\
 & \leq  NC'_y\big(\|\mathrm{div}_x(c_X(\cdot,y)\xi(\cdot,y)p(\cdot,y))\|_{L^2(E_{n'})} + \|u_n(\cdot,y)\|_{L^2(E_{n'})}\big). 
 \end{align}
 Raising both sides to the power $q$, integrating over $y \in D_n$ and raising to the $1/q$ gives
 \[\|u_n\|_{L^q(F_n)} + \|\nabla_x u_n\|_{L^q(F_n;\R^d)}\leq NC(\|\mathrm{div}_x(c_X\xi p)\|_{L^2(F_n')} + \|u_n\|_{L^2(F_{n'})}), \]
 where $C$ is a constant depending only on $C'$ and $q$. Since $\nabla_x u_n = \nabla_x \phi^*$ the proof will be complete as soon we establish that $\|u_n\|_{L^2(F_{n'})} < \infty$. To this end we use the Poincar\'e inequality \cite[Theorem~5.8.1]{evanspartial2010} to obtain
 \begin{align}
     \|u_n(\cdot,y)\|_{L^2(E_{n'})} & \leq C_{n'}\|\nabla_x u_n(\cdot,y)\|_{L^2(E_{n'})} = C_{n'}\|\nabla_x \phi^* (\cdot,y)\|_{L^2(E_{n'})} \leq C_{n'}\epsilon_{n'}^y\|\phi^*(\cdot,y)\|_{\Wcal(y)}
     \label{eqn:un_bound}
 \end{align}
 where $C_n'$ is the constant coming from the Poincar\'e inequality, which depends only on $E_{n'}$ and we recall that $\epsilon_{n'}^y = \sup_{x \in E_{n'}} 1/\{\lambda_{\min}(c_X(x,y))p(x,y)\}$. Next, note that by optimality of $\phi^*$, we have the bound $J^y(\phi^*(\cdot,y)) \leq J^y(0) = |\xi|^2_{\Wcal(y)}$, which together with \eqref{eqn:phi_wbound} implies that $\|\phi^*(\cdot,y)\|_{\Wcal(y)} \leq 2|\xi|_{\Wcal(y)}$.
 Combining this estimate with \eqref{eqn:un_bound}, squaring both sides, integrating over $y \in D_{n'}$ and taking the square root again leads to the estimate 
 \begin{equation} \label{eqn:un_L2bound} 
 \|u_n\|_{L^2(F_{n'})} \leq \tilde C\left(\int_F\xi(z)^\top c_X(z)\xi(z)p(z)dz\right)^{1/2},
 \end{equation} where $\tilde C = 2C_{n'}\sup_{z \in F_{n'}} 1/\{\lambda_{\min}(c_X(z))p(z)\}$.
 The right hand side of \eqref{eqn:un_L2bound} is finite courtesy of Assumption~\ref{ass:conditions}\ref{item:finite_growth} and \ref{item:divergence}, which completes the proof.
\end{proof}
With this technical lemma proved we can now establish that $\phi^*$ solves the variational problem.
\begin{proof}[Proof of Lemma~\ref{lem:variational}]
    We have the following estimates,
    \begin{align}
        & \inf_{\phi \in \Dcal}\int_F (\nabla_x \phi(z)-\xi(z))^\top c_X(z)(\nabla_x \phi(z) - \xi(z))p(z)dz = \inf_{\phi \in \Dcal} \int_D J^y(\phi)dy 
         \geq \int_D \inf_{\phi \in C^2(E)} J^y(\phi)dy \\
        &\qquad   \geq \int_D \inf_{\phi \in \Wcal_0(y)} J^y(\phi)dy 
         =\int_F (\nabla_x \phi^*(z)-\xi(z))^\top c_X(z)(\nabla_x \phi^*(z) - \xi(z))p(z)dz, \label{eqn:variational_lower_bound}
    \end{align} 
    where $\phi^*$ is the constructed optimum from Lemma~\ref{lem:variational_technical}. The first inequality followed because for any $\phi \in \Dcal$ we have that $\phi(\cdot,y) \in C^2(E)$ for a.e.\ $y \in D$ and the second inequality followed because the space $\Wcal_0(y) \cup C^2(E)$ is larger than $C^2(E)$, but the optimizer cannot be in $C^2(E)\setminus \Wcal_0(y)$ as $J(\phi) = +\infty$ for such functions. To obtain an upper bound recall that $\phi^*$ is jointly measurable in $x$ and $y$ and $\phi^*(\cdot,y) \in C^2(E)$ so it is itself a member of $\Dcal$. Hence
    \begin{align} &
        \inf_{\phi \in \Dcal} \int_F (\nabla_x \phi(z)-\xi(z))^\top c_X(z)(\nabla_x \phi(z) - \xi(z))p(z)dz \\
        & \qquad \qquad \leq \int_F (\nabla_x \phi^*(z)-\xi(z))^\top c_X(z)(\nabla_x \phi^*(z) - \xi(z))p(z)dz \label{eqn:variational_upper_bound}
    \end{align}
The two bounds \eqref{eqn:variational_lower_bound} and \eqref{eqn:variational_upper_bound} together establish optimality of $\phi^*$. The uniqueness statement follows from the unique equivalence class solution $[\phi^*(\cdot,y)]$ to the variational problem guaranteed by Lemma~\ref{lem:variational_technical}. 
\end{proof}

\subsection{Proof of Proposition~\ref{prop:worst_case}}

Next we turn our attention to the construction of the worst-case measure $\P^*$. Since we only have Sobolev regularity for $(x,y) \mapsto \nabla_x \phi^*(x,y)$, rather than say joint Lipschitz continuity or differentiability in $x$ and $y$, classic SDE stability and recurrence results are not available. Instead, the proof follows in a very similar fashion to \cite[Theorem~5.5]{itkin2025ergodic} using very recent results on generalized Dirichlet forms. We refer to the text \cite{leeanalytic2022} for an accessible presentation of the generalized Dirichlet form results and for any notation and terminology below that is not explicitly defined in this paper.

To carry out this program we introduce the symmetric Dirichlet form $(\Ecal^0,D(\Ecal^0))$ as the closure on $L^2(F,\mu)$ of
\[\Ecal^0(u,v) = \int_F \nabla u(z)^\top c(z)\nabla v(z)p(z)dz; \qquad u,v \in C_c^\infty(F),\]
where $L^2(F,\mu)$ is the space of square-integrable functions with respect to the measure $d\mu(z) = p(z)dz$. The generator $(L^0,D(L^0))$ corresponding to this Dirichlet form is readily seen, via integration by parts, to satisfy
\[L^0u  = \frac{1}{2}\mathrm{Tr}(c\nabla^2u) + (\ell^0)^\top c \nabla u; \qquad u \in C_c^\infty(F),\]
where $\ell^0 = \frac{1}{2}(c)^{-1}\mathrm{div} \, c + \frac{1}{2}\nabla \log p$. Next we define the quantity
\[\beta(z) = \begin{pmatrix}c_X(z)\nabla_x \phi^*(z) \\ c_Y(z)b_Y(z) \end{pmatrix} - c(z)\ell^0(z) = \begin{pmatrix}
    c_X(z)(\nabla_x \phi^*(z) - \ell_X(z)) \\ -c_Y(z)\ell_Y(z)
\end{pmatrix},\]
where we recall $\ell_X$ and $\ell_Y$ given by \eqref{eqn:ell_X} and \eqref{eqn:ell_Y} respectively.
The vector field $\beta $ will play the role of a $\mu$-divergence free perturbation to the symmetric Dirichlet form. The construction of the measures $(\P^*_z)_{z \in F}$ crucially relies on the following technical lemma.

\begin{lem} \label{lem:L1_semigroup}
    There exists an operator $(L,D(L))$ on $L^1(F,\mu)$ such that the following hold:
    \begin{enumerate}
        \item $C_c^\infty(F) \subset D(L)$ and
        \[Lu = L^0u + \beta^\top \nabla u; \qquad u \in C_c^\infty(F),\]
    \item For every bounded $u \in D(L)$ and every compactly supported and bounded $v \in W^{1,2}(F)$ we have
    \begin{equation} \label{eqn:IBP}
        \Ecal^0(u,v) - \int_F v(z)\beta(z)^\top \nabla u(z)p(z)dz =- \int_F v(z)Lu(z)p(z)dz,
    \end{equation}
    \item $(L,D(L))$ generates a strongly continuous contraction semigroup $(T_t)_{t \geq 0}$ on $L^1(F,\mu)$. Moreover, $T_tf$ has a continuous version $P_tf$ for every $f \in \Bcal_b(F)$ and $t > 0$. \label{item:semigroup}
        \end{enumerate}
\end{lem}
\begin{proof}
    The first two items of the lemma will follow from \cite[Theorem~1.5]{stannatdirichlet1999} and the last item from \cite[Theorem~2.31]{leeanalytic2022}(which is applicable courtesy of \cite[Remark~2.40]{leeanalytic2022}) as soon as we verify that $\beta \in L^q_{\mathrm{loc}}(F;\R^{d+m})$ for some $q > d+m$ and that
    \begin{equation} \label{eqn:divergence_free} 
    \int_F (L^0u(z) + \beta(z)^\top \nabla u(z) )p(z)dz = 0; \qquad \forall u \in C_c^\infty(F).
    \end{equation} We deduce from Assumption~\ref{ass:inputs} on $c\,, p$ and $b_Y$ together with the $L^q_{\mathrm{loc}}(F;\R^d)$ result for $\nabla_x \phi^*$ guaranteed by Lemma~\ref{lem:variational} that $\beta \in L^q_{\mathrm{loc}}(F;\R^{d+m})$ for every $q \in [2,\infty)$. In particular this holds for $q > d+m$. Next a direct calculation shows that $L^0 u(z)p(z) = \mathrm{div}(c(z) \nabla u(z) p(z))$ for every $u \in C_c^\infty(F)$. The divergence theorem, in turn, yields $\int_F L^0 u(z) p(z)dz = 0$. To handle the perturbation term we similarly use the divergence theorem,
    \begin{align}
        \int_F \beta(z)^\top \nabla u(z)& p(z)dz = -\int_F \mathrm{div}(\beta(z)p(z))u(z)dz  \\ & =
        -\int_F \Big(\mathrm{div}_x\big(c_X(z)(\nabla_x\phi^*(z)-\ell_X(z))p(z)\big) - \mathrm{div}_y(c_Y(z)\ell_Y(z)p(z))\Big)u(z)dz = 0,
    \end{align}
    where in the last equality we used \eqref{eqn:Euler-Lagrange}. This establishes \eqref{eqn:divergence_free} and completes the proof.
\end{proof} 
A consequence of Lemma~\ref{lem:L1_semigroup} is the existence of a diffusion process with semigroup $(P_t)_{t > 0}$ given in Lemma~\ref{lem:L1_semigroup}\ref{item:semigroup}. Formally we augment the state space $F$ with a cemetery state $\Delta$ by letting $F_\Delta = F \cup \{\Delta\}$ be the one-point compactification of $F$. Next we introduce the measurable space $(\Omega_\Delta,\Fcal_\Delta)$ given by
\[\Omega_\Delta = \{\omega \in C([0,\infty);F_\Delta): \omega_{t+h} = \Delta \text{ if } \omega_t = \Delta \text{ for all } h,t\geq 0\}\] 
and $\Fcal_\Delta$ being the Borel $\sigma$-algebra induced by the topology of local uniform convergence. With a slight abuse of notation we denote by $Z$ the coordinate process on this space. Then we have the following existence result
\begin{lem} \label{lem:diffusion_existence}
    There exists a diffusion  
    \[\mathbb{M} = (\Omega_\Delta,\Fcal_\Delta,(\Fcal_t)_{t \geq 0},(Z_t)_{t \geq 0},(\P^*_z)_{z \in F_\Delta})\]
    with state space $F$, lifetime $\zeta := \inf\{t \geq 0:Z_t = \Delta\}$ and transition semigroup $(P_t)_{t \geq 0}$ of Lemma~\ref{lem:L1_semigroup}\ref{item:semigroup}. That is for every $t \geq 0, \, z \in F$ and $f \in \Bcal_b(F)$ it holds that $P_tf(z) =  \E^*_z[f(Z_t)]$ where $\E^*_z[\cdot]$ denotes expectation under $\P^*_z$.
\end{lem}
The proof of Lemma~\ref{lem:diffusion_existence} follows in exactly the same way as the proof of \cite[Lemma~C.3]{itkin2025ergodic}, as a consequence of Lemma~\ref{lem:L1_semigroup}, and is hence omitted. Next, we show that $\mathbb{M}$ is a global weak solution to \eqref{eqn:worst_case_SDE} and is ergodic with invariant measure $\mu$.

\begin{lem} \label{lem:diffusion_properties}
    The process $\mathbb{M}$ of Lemma~\ref{lem:diffusion_existence} is strictly irreducible, recurrent, nonexplosive and ergodic with invariant measure $\mu$. Moreover, $\mathbb{M}$ is a weak solution to $\eqref{eqn:worst_case_SDE}$ and \eqref{eqn:ergodic} holds for every locally bounded $h \in L^1(F,\mu)$.
\end{lem}
\begin{proof} The strict irreducibility claim is a consequence of \cite[Proposition~2.39]{leeanalytic2022}. Next, we will prove recurrence using the criteria developed in \cite{gim2018recurrence} together with Assumption~\ref{ass:conditions}\ref{item:test_functions}. Indeed, Assumption~\ref{ass:conditions}\ref{item:test_functions} ensures that there exist functions $\chi_n \in C_c^\infty(F)$ with $0 \leq \chi_n \leq 1$, converging pointwise to one and such that
\begin{align} 
\lim_{n \to \infty} &  \Ecal^0(\chi_n,\chi_n) = \int_F \nabla \chi_n(z)^\top c(z)\nabla \chi_n(z)p(z)dz \\
& \leq \lim_{n \to \infty}2\left(\int_F \nabla_x \chi_n(z)^\top c_X(z)\nabla_x\chi_n(z)p(z)dz + \int_F \nabla_y \chi_n(z)^\top c_Y(z)\nabla_y\chi_n(z)p(z)dz\right)  = 0,
\end{align}
where we used the inequality $v^\top c(z) v \leq 2(v_X^\top c_X(z)v_X + v_Y^\top c_Y(z)v_Y)$ for any $v \in \R^{d+m}$, which holds since $c(z)$ is positive definite.
Next, we use Cauchy--Schwarz to deduce that
\begin{align}
     & \int_F|\beta^\top \nabla \chi_n(z)|p(z)dz  = \int_F|(\nabla_x \phi^*(z) - \ell_X(z))^\top c_X(z)\nabla_x \chi_n(z) - \ell_Y(z)^\top c_Y(z) \nabla_y \chi_n(z)\big|p(z)dz \\
      \leq & \bigg(\!\int_F (\nabla_x \phi^*(z)-\ell_X(z))^\top c_X(z)(\nabla_x \phi^*(z) - \ell_X(z))p(z)dz\!\bigg)^{\!\frac{1}{2}}\!\bigg(\!\int_F \nabla_x \chi_n(z)^\top c_X(z)\nabla_x \chi_n(z)p(z)dz\!\bigg)^\frac{1}{2} \\
     & \quad + \left(\int_F \ell_Y(z)^\top c_Y(z)\ell_Y(z)p(z)dz\right)^\frac{1}{2}\left(\int_F\nabla_y \chi_n(z)^\top c_Y(z)\nabla_y \chi_n(z)p(z)dz\right)^\frac{1}{2}. \label{eqn:chi_n_bound}
\end{align}
Since $\phi^*$ is a minimizer of the variational problem \eqref{eqn:variational} we deduce from the fact that $0 \in \Dcal$, the bound $\int_F\nabla_x \phi^*(z)^\top c_X(z)\nabla_x \phi^*(z)p(z)dz \leq \int_F \xi(z)^\top c_X(z)\xi(z)p(z)dz$. Together with the integrability bounds of Assumption~\ref{ass:conditions}\ref{item:finite_growth} and \ref{item:divergence} we obtain that the integral terms on the right hand side of \eqref{eqn:chi_n_bound} are finite. Assumption~\ref{ass:conditions}\ref{item:test_functions} then yields the convergence to zero of the terms in \eqref{eqn:chi_n_bound} as $n \to \infty$. In summary, we have shown that
\[\lim_{n \to \infty} \left(\Ecal^0(\chi_n,\chi_n) + \int_F|\beta^\top(z)\nabla \chi_n(z)|p(z)dz\right)  = 0.\]
Remark~15 and Corollary~8(b) in \cite{gim2018recurrence} now yield recurrence of the semigroup $(P_t)_{t \geq 0}$, which implies recurrence of $\mathbb{M}$ as $(P_t)_{t \geq 0}$ is its semigroup. A consequence of recurrence is nonexplosivity of $\mathbb{M}$, which in this context can be deduced from \cite[Corollary~3.23]{leeanalytic2022} together with conservativity of $(P_t)_{t \geq 0}$, which is an immediate consequence of its recurrence. That $\mu$ is an ergodic measure for $\mathcal{M}$ and the ergodic property \eqref{eqn:ergodic} holds now follows in exactly the same way as in the proof of \cite[Lemma~C.5]{itkin2025ergodic}. That $\mathbb{M}$ is a weak solution to \eqref{eqn:worst_case_SDE} follows via standard arguments connecting the process $\mathbb{M}$ to the martingale problem for the generator $L$ via \eqref{eqn:IBP} and using the well known equivalence between martingale problems and weak solutions of SDEs. This is precisely the result \cite[Theorem~3.22(i)]{leeanalytic2022}, which we obtain here by following the proof of \cite[Chapter~3]{leeanalytic2022} verbatim from Proposition~3.12 onwards, but in our setting of a general open domain $F$ rather than all of $\R^{d+m}$. This completes the proof.   \end{proof}
We are now ready to show that $\P^*_z$ is in $\Pcal_0$.
\begin{proof}[Proof of Proposition~\ref{prop:worst_case}]
The process $\mathbb{M}$ of Lemma~\ref{lem:diffusion_existence} has been shown to satisfy \eqref{eqn:opt_wealth_easy} and the ergodic property \eqref{eqn:ergodic} in Lemma~\ref{lem:diffusion_properties}. As such, by definition of $\Pcal$,  its law $\P^*_z \in \Pcal$ for every $z \in F$.

It now remains to show the existence of a growth-optimal portfolio with finite asymptotic growth rate to deduce its inclusion in $\Pcal_0$. The discussion at the beginning of Section~\ref{sec:main_results} yields that \eqref{eqn:bP_condition} must hold where we replace $\P$ with $\P^*_z$ and $b_X^\P$ with $\nabla_x \phi^*(Z)$. But using the ergodic property and that $\phi^*$ is a minimizer for the variational problem \eqref{eqn:variational} we see that 
\[\lim_{T \to \infty} \frac{1}{T}\int_0^T \nabla_x \phi^*(Z_t)^\top c_X(Z_t)\nabla_x \phi^*(Z_t)dt = \int_F \nabla_x \phi^*(z)^\top c_X(z)\nabla_x \phi^*(z)p(z)dz < \infty; \qquad \P^*_z\text{-a.s.}\]
This shows that \eqref{eqn:bP_condition} holds and, in fact, that
\begin{equation} \label{eqn:worst_case_growth}
    \sup_{\theta \in \Theta} g(\theta;\P^*_z) = g(\theta^*;\P^*_z) = \frac{1}{2}\int_F \nabla_x \phi^*(z)^\top c_X(z)\nabla_x \phi^*(z)p(z)dz.
\end{equation}
This establishes that $\P^*_z \in \Pcal_0$ and completes the proof.
\end{proof}

\subsection{Proof of Theorem~\ref{thm:main}} \label{sec:proof_main}
It now just remains to prove the main result Theorem~\ref{thm:main}. To accomplish this we need the following lemma.
\begin{lem}  \label{lem:approx}
    There exist functions $ \phi_n \in C_c^\infty(F)$ such that 
    \[\lim_{n \to \infty} \int_F (\nabla_x \phi_n(z)-\nabla_x \phi^*(z))^\top c_X(z)(\nabla_x \phi_n(z) - \nabla_x \phi^*(z))p(z)dz = 0\]
\end{lem}
\begin{proof}
First we extend $\phi^*$ to all of $\R^{d+m}$ by setting $\phi^*(z) = 0$ if $z \not \in F$. Next let $(\eta_\epsilon)_{\epsilon > 0}$ be a standard mollifier on $\R^{d+m}$ and define $\phi^n = \phi^**\eta_{1/n}$. Clearly $\phi^n \in C_c^\infty(\R^{d+m})$ and so its restriction to $F_k$ belongs to $W^{1,2}(F_k)$ for every $k \in \N$. By density of $C_c^\infty(F_k)$ in $W^{1,2}(F_k)$ we can find functions $\phi^n_{j,k} \in C_c^\infty(F_k)$ such that
\begin{equation}
    \label{eqn:Fk_converge}
    \lim_{j \to \infty} \|\phi^n_{j,k} - \phi^n\|_{W^{1,2}(F_k)} = 0.
\end{equation}
Next, to simplify the exposition for any open set $U \subset F$ and function $\phi: F \to \R$ weakly differentiable in $x$ we will write
\[\|\phi\|_{\Wcal(U)} = \int_U \nabla_x \phi(z)^\top c_X(z)\nabla_x\phi(z)p(z)dz.\] Then it follows from \eqref{eqn:Fk_converge} that
\begin{align} 
\lim_{j \to \infty} \|\phi^n_{k,j} - \phi^n\|_{\Wcal(F_k)} \leq \sup_{z \in F_k}\lambda_{\max}(c_X(z))p(z)\lim_{j \to \infty}\|\nabla_x\phi^n_{j,k}-\nabla_x\phi^n\|_{L^2(F_k;\R^d)} = 0.
\end{align}  Next, note that because $\phi^*$ is itself continuously differentiable in $x$, with $\nabla_x\phi^* \in L^2(F_k;\R^d)$ courtesy of Lemma~\ref{lem:variational_technical}, we have for every $i=1,\dots,d$ that $\partial_{x_i} \phi^n = (\partial_{x_i}\phi^*)* \eta_{1/n}$, which by standard properties of mollification (see e.g. \cite[Theorem~~C.5.7]{evanspartial2010}) converges to $\partial_{x_i}\phi^*$ as $n \to \infty$ in $L^2(F_k)$. It then follows, again using uniform boundedness of $c_Xp$ on $F_k$, that
$\lim_{n \to \infty}\|\phi^n - \phi^*\|_{\Wcal(F_k)} = 0.$
Finally, note by monotone convergence, that $\lim_{k \to \infty} \|\phi^*\|_{\Wcal(F\setminus F_k)} = 0$. As such given a tolerance $\epsilon > 0$ we first choose $k$, then $n = n(k)$ and then $j = j(n,k)$ large enough so that
\[\|\phi^*\|_{\Wcal(F\setminus F_k)} < \frac{\epsilon}{3}, \qquad \|\phi^* - \phi^n\|_{\Wcal(F_k)} < \frac{\epsilon}{3}, \qquad \|\phi^n - \phi^n_{j,k}\|_{\Wcal(F_k)} < \frac{\epsilon}{3}.\]
Then we see by the triangle inequality that
\begin{align}
\|\phi^* - \phi^n_{j,k}\|_{\Wcal(F)} & = \|\phi^*\|_{\Wcal(F\setminus F_k)} + \|\phi^* - \phi^n_{j,k}\|_{\Wcal(F_k)} \\
& \leq \|\phi^*\|_{\Wcal(F\setminus F_k)} + \|\phi^* - \phi^n\|_{\Wcal(F_k)} + \|\phi^n -  \phi^n_{j,k}\|_{\Wcal(F_k)} < \epsilon.
\end{align}
Since each $\phi^n_{j,k} \in C_c^\infty(F)$ this completes the proof.
\end{proof}

\begin{proof}[Proof of Theorem~\ref{thm:main}] We set $\theta^n_t = \nabla_x \phi_n(Z_t)$, where $\phi_n$ is as in Lemma~\ref{lem:approx}. Since $\phi_n \in C_c^\infty(F)$ we can apply It\^o's formula and the steps in Section~\ref{sec:heuristic} to obtain for any $\P \in \Pcal$ that
\begin{align}
    g(\theta^n;\P) & = \frac{1}{2}\int_F \xi(z)^\top c_X(z)\xi(z)p(z)dz - \frac{1}{2}\int_F (\nabla_x \phi_n(z)-\xi(z))^\top c_X(z)(\nabla_x \phi_n(z) - \xi(z))p(z)dz \\
    & = \int_F \nabla_x\phi_n(z)^\top c_X(z)\xi(z)p(z)dz - \frac{1}{2}\int_F \nabla_x\phi_n(z)^\top c_X(z)\nabla_x \phi_n(z)p(z)dz \\
    & = \int_F \nabla_x\phi_n(z)^\top c_X(z)\nabla_x \phi^*(z)p(z)dz - \frac{1}{2}\int_F \nabla_x\phi_n(z)^\top c_X(z)\nabla_x \phi_n(z)p(z)dz,
\end{align} 
where in the final equality we used that $\phi^*$ is a weak solution to the Euler-Lagrange equation \eqref{eqn:Euler-Lagrange_xi}. This leads us to the lower bound
\begin{align} 
\lambda_{\Pcal} \geq \lim_{n \to \infty} \inf_{\P \in \Pcal} g(\theta^n;\P) = \frac{1}{2}\int_F \nabla_x \phi^*(z)^\top c_X(z)\nabla_x \phi^*(z)p(z)dz,
\end{align}
where we used Lemma~\ref{lem:approx} to compute the limit. 

To obtain the upper bound we use the measure $\P^*$ constructed in Proposition~\ref{prop:worst_case} (with any initial value $z \in F$). As shown in the proof of Proposition~\ref{prop:worst_case}, $\theta^*$ is growth-optimal under $\P^*$ with asymptotic growth rate derived in \eqref{eqn:worst_case_growth}. Hence,
\[\lambda_{\Pcal} \leq  \sup_{\theta \in \Theta} g(\theta;\P^*) = g(\theta^*;\P^*) = \frac{1}{2}\int_F \nabla_x \phi^*(z)^\top c_X(z)\nabla_x \phi^*(z)p(z)dz.\]
This establishes the robust growth rate formula so it now just remains to show that $\theta^*$ achieves this same asymptotic growth rate under every $\P \in \Pcal_0$. To this end we fix $\P \in \Pcal_0$ and note that for any $\theta \in \Theta$ we have
\[\log V^\theta_T = \int_0^T \theta_t^\top dX_t - \frac{1}{2}\int_0^T \theta_t^\top c(Z_t)\theta_t dt.\] Now taking $\theta^*$ and $\theta_n$ from the first part of this proof we see that
\begin{align} \log V_T^{\theta^*} = \log V_T^{\theta_n}    - \frac{1}{2}\int_0^T (\nabla_x \phi^*)^\top c_X\nabla_x \phi^*(Z_t)dt & + \frac{1}{2} \int_0^T \nabla_x \phi_n^\top c_X\nabla_x  \phi_n(Z_t)dt \qquad \label{eqn:opt_wealth_easy} \\ & \qquad   + \int_0^T(\theta^*_t - \theta^n_t)^\top dX_t \label{eqn:opt_wealth_hard}.
\end{align}
We now divide by $T$ and send $T \to \infty$ and then $n \to \infty$. The growth rate invariance property of $\theta_n$ together with the ergodic property \eqref{eqn:ergodic} and Lemma~\ref{lem:approx} shows that the terms on the right hand side of \eqref{eqn:opt_wealth_easy} converge $\P$-a.s.\ to $\lambda_{\Pcal}$. To complete the proof it suffices to show that
\begin{equation} \label{eqn:stoch_int_growth}
  \lim_{n \to \infty}\sup\Big\{\gamma \in \R: \liminf_{T \to \infty} \frac{1}{T}\int_0^T(\theta^*_t - \theta_t^n)^\top dX_t\geq \gamma \Big\} = 0.
\end{equation}
To this end we recall the dynamics of $X$ under $\P$, which leads us to the estimates
\begin{equation} \label{eqn:stoch_int_bound}
\left|\frac{1}{T}\int_0^T (\theta^*_t - \theta^n_t)^\top dX_t\right| \leq\left|\frac{1}{T}\int_0^T (\nabla_x\phi^*(Z_t)-\nabla_x \phi_n(Z_t))^\top c_X(Z_t) b^\P_{X,t}dt\right| +\frac{1}{T}|L^n_T|, 
\end{equation}
where $L^n_T := \sum_{i=1}^d\sum_{j=1}^{d+m}\int_0^T (\partial_i\phi^*(Z_t)-\partial_i \phi_n(Z_t)) c^{1/2}_{ij}(Z_t)dW_{j,t}$ is a local martingale. We have that 
\[\lim_{T \to \infty}\frac{1}{T}[L^n]_T = \int_F (\nabla_x \phi^*(z) - \nabla_x \phi_n(z))^\top c_x(z)(\nabla_x \phi^*(z) -\nabla_x \phi_n(z))p(z)dz < \infty,\]
so by \cite[Lemma~1.3.2]{fernholz2002stochastic}, $\lim_{T \to \infty} \frac{1}{T}L^n_T = 0$. For the drift term we use Cauchy--Schwarz to obtain
\begin{align}\bigg|\frac{1}{T}\int_0^T & (\nabla_x\phi^*(Z_t)-\nabla_x \phi_n(Z_t))^\top c_X(Z_t) b^\P_{X,t})dt\bigg| \\&  \leq \left(\frac{1}{T}\int_0^T(\nabla_x \phi^* - \nabla_x \phi_n)^\top c_X(\nabla_x \phi^*-\nabla_x \phi_n)(Z_t)dt\right)^{\frac{1}{2}} \left(\frac{1}{T}\int_0^T (b_{X,t}^\P)^\top c_X(Z_t)b_{X,t}^\P dt\right)^{\frac{1}{2}} \label{eqn:Pcal0_estimate}.
\end{align}
By the ergodic property \eqref{eqn:ergodic} and Lemma~\ref{lem:approx} we have that the first term on the right hand side of \eqref{eqn:Pcal0_estimate} tends $\P$-a.s.\ to zero as $T$ and then $n$ goes to infinity. The condition \eqref{eqn:bP_condition}, which holds here since $\P \in \Pcal_0$, ensures that the second term in \eqref{eqn:Pcal0_estimate} remains finite on a set of strictly positive measure when sending $T \to \infty$. This establishes \eqref{eqn:stoch_int_growth} and completes the proof. 
\end{proof}

\bibliographystyle{plain}
\bibliography{References}

\end{document}